\RequirePackage{fix-cm}
\RequirePackage{fixltx2e}
\documentclass[12pt]{article}
\usepackage[latin9]{inputenc}
\usepackage[a4paper]{geometry}
\usepackage{float} 
\usepackage{rotfloat}
\usepackage{enumitem}
\usepackage{booktabs}
\usepackage{amsmath} 
\usepackage{amsthm}
\usepackage{amssymb}
\RequirePackage{amsfonts}
\usepackage{graphicx}
\usepackage[authoryear]{natbib}
\usepackage[unicode=true,pdfusetitle,
 bookmarks=true,bookmarksnumbered=false,bookmarksopen=false,
 breaklinks=true,pdfborder={0 0 0},backref=false,colorlinks=false]
 {hyperref}
\usepackage{breakurl}
\usepackage{xr}
\makeatletter

\usepackage{titlesec}
\titleformat{\section}[block]{\bfseries\filcenter}{\thesection.}{1em}{}
\titleformat{\subsection}[hang]{\bfseries}{\thesubsection.}{1em}{}

\providecommand{\tabularnewline}{\\}
\floatstyle{ruled}
\newfloat{algorithm}{tbp}{loa}
\providecommand{\algorithmname}{{\footnotesize Algorithm}}
\floatname{algorithm}{\footnotesize{\protect\algorithmname}}
 
  \theoremstyle{definition}
  \newtheorem{defn}{\protect\definitionname}
\theoremstyle{plain}
\newtheorem{thm}{\protect\theoremname}
  \theoremstyle{plain}
  \newtheorem{cor}{\protect\corollaryname}

\@ifundefined{date}{}{\date{}}

\usepackage{Tabbing}
\newcommand{\ra}[1]{\renewcommand{\arraystretch}{#1}}

\makeatother

\providecommand{\definitionname}{Definition}
\providecommand{\corollaryname}{Corollary}
\providecommand{\theoremname}{Theorem}

\begin{document}
\global\long\def\az{\mathcal{A}^{0}}
\global\long\def\as{\mathcal{A}^{*}}
\global\long\def\ak{\mathcal{A}^{k}}
\global\long\def\ddd{,\ldots,}
\global\long\def\ba{\mathcal{\boldsymbol{A}}}
\global\long\def\bo{\boldsymbol{\omega}}
\global\long\def\aa{\mathcal{A}}

\title{\textbf{Sparsity Oriented Importance Learning for High-dimensional Linear Regression}}
\author{Chenglong Ye \thanks{School of Statistics, University of Minnesota (yexxx323@umn.edu)},
Yi Yang \thanks{Department of Mathematics and Statistics, McGill University (yi.yang6@mcgill.ca)},
Yuhong Yang \thanks{School of Statistics, University of Minnesota (yangx374@umn.edu)}}
\date{\today}
\maketitle
\begin{abstract}

With now well-recognized non-negligible model selection uncertainty, data analysts should no longer be satisfied with the output of a single final model from a model selection process, regardless of its sophistication. To improve reliability and reproducibility in model choice, one constructive approach is to make good use of a sound variable importance measure. Although interesting importance measures are available and increasingly used in data analysis, little theoretical justification has been done. In this paper, we propose a new variable importance measure, {\bf s}parsity {\bf o}riented {\bf i}mportance {\bf l}earning (SOIL), for high-dimensional regression from a sparse linear modeling perspective by taking into account the variable selection uncertainty via the use of a sensible model weighting. The SOIL method is theoretically shown to have the inclusion/exclusion property: When the model weights are properly around the true model, the SOIL importance can well separate the variables in the true model from the rest.  In particular, even if the signal is weak, SOIL rarely gives variables not in the true model significantly higher important values than those in the true model. Extensive simulations in several illustrative settings and real data examples with guided simulations show desirable properties of the SOIL importance in contrast to other importance measures.
\end{abstract}

\noindent {\bf Keywords:} Variable importance; Model averaging; Adaptive regression
by mixing; Reliability and reproducibility.

\section{INTRODUCTION}

Variable importance has been an interesting research topic that helps
to identify which variables are most important for understanding, interpretation, estimation
or prediction purposes. The potential usages of variable importance measures include: 1. They help reduce the list of variables to be considered by screening out those with importance values below a threshold. This leads to cost and time saving in data analysis; 2. They also
help decision makers to obtain a more comprehensive understanding of
the underlying data generation process than trusting any single model by a variable selection procedure; 3. They offer a ranking of variables that can be used to consider model selection or model averaging in a nested fashion, which simplifies the consideration of all subset models;  4. They can help decision makers to change or replace variables based on practical considerations. See \citealp{feldman2005relative,louppe2013understanding,braun2011exploratory,gromping2015variable,hapfelmeier2014new,archer2008empirical,strobl2007bias}
for reference.

Under the linear regression setting, various
methods have been proposed for evaluating variable importance. The first type includes simple measures based on a final selected model, e.g., $t$-test values, (standardized)
regression coefficients, and $p$-values of the variables. This approach has the severe drawback associated with any ``winner takes all" variable selection method. The variable selection uncertainty is totally ignored
and all the non-selected variables have zero importance. 

Another approach is based on
the $R^{2}$ decomposition. \citet{lindeman1980introduction} used
the improved explained variance averaged over all possible orderings
of predictors to provide a ranking of the predictors. \citet{feldman1999proportional}
extended it to the weighted version (PMVD). Several encouraging methods,
such as dominance analysis \citep{budescu1993dominance}, hierarchical
partitioning \citep{chevan1991hierarchical}, information criterion
based method \citep{theil1988information} and the product of standardized
true coefficients and partial correlation \citep{hoffman1960paramorphic},
have also been proposed. 

Besides importance measuring with parametric models, nonparametric approaches are also available. For regression and classification, random forest \citep{breiman2001random} and
its variants have attracted a lot of attention in many fields. \citet{breiman2001random}
proposed two versions of variable importances for random forest. \citet{ishwaran2007variable}
studied the theoretical properties of variable importance for binary
regression with random forest. There, the variable importance is defined as the difference between the prediction error before and after the variable is noised up. Under proper assumptions, the variable importance is shown to converge and suitably upper-bounded. \citet{strobl2008conditional} proposed
conditional variable importance for random forest to correct the bias
of variable importance when there exist correlated
variables. \cite{ferrari2015confidence} assess variable importance from a variable selection confidence set (VSCS) perspective.

In this paper, we propose a sparsity oriented importance
learning (SOIL) for high-dimensional regression data. For our
approach, by assigning weights to the candidate linear models (or generalized linear models for classification), we come up with measures of importance of the predictors in an absolute scale in $[0,1]$. 

Several features/advantages of our method can be concluded as follows. First,
it involves multiple high-dimensional variable selection methods and combines all
their solution path models, which produces many candidate models rather than being based on only one model selection method. The resulting importance values are thus more reliable than trusting one method alone. Second, SOIL uses external weighting, which
is independent of the model selection methods. This can avoid possible bias brought up by using a method both for coming up with candidate models and for assessing the models for weighting. Third, from the main
theorem in the paper, we gain a theoretical understanding of our method.
We prove that the importances of the variables will tend to either
0 or 1 as the sample size increases, as long as the weighting is sensible. Last but not least, compared with
other importance measures, our method also shows excellent performances in the numerical study, with desirable behaviors such as {\it exclusion, inclusion,
order preserving, robustness,} etc. 

In the current era of rich high-dimensional data, with the well-recognized severe problem of irreproducibility of scientific findings (see, \citealp[e.g.][]{ioannidis2011improving, mcnutt2014raising, stodden2015reproducing}), we believe the use of informative importance measures can much improve the reliability of data analysis in multiple ways: 
\begin{enumerate}
\item First, if the data analyst has already chosen a set of covariates for finalizing a model to be recommended, the SOIL importance measure is helpful to put the model under a more objective light. He/she can immediately inspect if some variables deemed important by SOIL are missing in the set or the other way around. If so, the analyst may want to investigate on the matter. For instance, residuals from the model based on the current set of covariates, when plotted against the missing variables, may reveal their relevance. Models with/without the variables in questions can be fit and compared for a better understanding on their usefulness. 

\item Based on the theoretical properties of the SOIL, variables most suitable for sparse modeling receive higher importance values. Thus the SOIL can be naturally used to find the best model for the data. In theory, any fixed cutoff in $(0,1)$ leads to a good performance (see Theorem 2). But the best cutoff depends on the purpose of the final model: for prediction accuracy, the cutoff should be lower and for identifying variables than can be validated at similar sample sizes in future studies, the cutoff should be higher. See e.g., \citet{Yang2005bm} to understand the subtle matter of the conflict between model identification and estimation/prediction.    

\item Whether one comes up with a set of covariates based on SOIL importance (as described above) or not (e.g., using a penalized likelihood based model selection method), the SOIL importance values of the variables help the data analyst get a sense on model selection uncertainty. More specifically, if there are quite a few variables having importance values similar to some in a final model (obtained from a trustworthy process that has, at least reasonably, justified the usefulness of the selected covariates, e.g., based on cross validation), it may indicate that the model selection uncertainty is perhaps high for the data and there are alternative choices of variables that can give similar predictive performances. In such a case, it is advantageous for the data analyst and the decision maker to be well-informed on possible alternative models/covariates to be used. For instance, if some covariates are much less costly for future experiments or operations, they may be preferred to be included in the final model even if their importance values are slightly lower than some other ones in a good model.

\item When estimating the regression function or prediction is the main goal, the understanding on degree of model selection uncertainty, together with other model selection diagnostic tools (see, e.g., \citealp{nan2014variable} for references), can help the data analyst decide on the choice between model selection and model averaging (see, \citealp{yang2003regression,chen2007model} for results on comparison between model selection and model averaging).          

\end{enumerate}  

In summary, the SOIL method is helpful in different stages of model building. It can be used to narrow down the set of covariates for further consideration and for reaching a final model with sound considerations. Equally or even more importantly, it provides an objective view on reliability of the model and the model selection uncertainty. This gives information unavailable in the traditional practice of glorifying the final model and thus can help much improve reproducibility of data analysis that involves variable selection.   

The remainder of the paper is organized as follows. In Section 2,
we introduce the proposed SOIL methodology and
provide a theoretical understanding on some key aspects. Section 3 presents
the details of choosing the candidate models and the weighting for
SOIL in practice. In Section 4, we conduct several simulations that fairly and informatively 
compare the performance of SOIL and three existing and commonly used variable
importance measures (LMG and two versions of random forest importances).
Furthermore, we apply these methods to two real datasets. A discussion
about variable importance is then presented in Section 5, followed
by the proofs of the results in Appendix.

\section{GENERAL METHODOLOGY}

In this section, we introduce the \emph{Sparsity Oriented Importance
Learning} (SOIL) procedure, which provides an objective and informative
profile of variable importances for high dimensional regression and
classification models. We consider the regression setting first, and
the generalization to the classification model will be discussed later
in Section \ref{sub:Estimating-the-weights}.

Let $\mathbf{X}=(X_{1}\ddd X_{p})$ be the $n\times p$ design
matrix with $X_{j}=(x_{1j}\ddd x_{nj})^{\intercal}$, $j=1\ddd p$,
and $\mathbf{y}=(y_{1}\ddd y_{n})^{\intercal}$ be the $n$-dimensional
response vector. The design matrix can also be written as $\mathbf{X}=(\mathbf{x}_{1}\ddd\mathbf{x}_{n})^{\intercal}$,
where $\mathbf{x}_{i}=(x_{i1}\ddd x_{ip})^{\intercal}$, $i=1\ddd n$.
We consider the following underlying linear regression model
\[
\mathbf{y}=\mathbf{X}\boldsymbol{\beta}^{*}+\boldsymbol{\varepsilon},
\]
where $\boldsymbol{\varepsilon}$ is the vector of $n$ independent
errors and $\boldsymbol{\beta}^{*}=(\beta_{1}^{*}\ddd\beta_{p}^{*})^{\intercal}$
is a $p$-dimensional vector of the true underlying model that generates
the data. In general, predictors may include those created by the
original predictors observed, such as $\sqrt{X_{1}}$, $X_{1}^{2}$
and $X_{1}X_{3}$. We adopt the sparsity assumption that most regression
coefficients $\beta_{j}^{*}$ are zero. Denote by $|\cdot|$ the cardinality
of a set. We assume $\boldsymbol{\beta}^{*}$ is $r^{*}$-sparse,
where $r^{*}=|\mathcal{A}^{*}|$ with $\mathcal{A}^{*}\equiv\mathrm{supp(\boldsymbol{\beta}^{*})}=\{j:\beta_{j}^{*}\neq0\}$.

SOIL importance depends on two ingredients: a manageable set of models
(often based on a preliminary analysis) and a reliable external weighting
method on the models. Together they can provide valuable information
on importance of the predictors.
 
Suppose that one can obtain a collection of models $\ba=\{\mathcal{A}_{k}\}_{k=1}^{K}$,
which can be either a full list of all-subset models when $p$ is
small, or a group of models obtained from high-dimensional variable selection procedures
such as Lasso \citep{tibshirani1996regression}, Adaptive Lasso \citep{zou2006adaptive}, SCAD \citep{fan2001variable} and MCP \citep{zhang2010nearly} etc., when $p$ is large. We refer to
$\mathcal{A}_{k}$, $k=1\ddd K$ as \emph{candidate models}, and $\mathbf{w}=(w_{1}\ddd w_{K})^{\intercal}$
as the corresponding weighting vector, which is estimated from the
data.

Given the set $\ba$ and the weighting $\mathbf{w}$, we define the
SOIL importance measure for the $j$-th variable, $j\in\{1\ddd p\}$,
as the accumulated sum of weights of the candidate models $\ak$ that
contains the $j$-th variable. That is
\[
\mbox{SOIL Importance}:\quad S_{j}\equiv S(j;\mathbf{w},\ba){\displaystyle ={\textstyle \sum_{k=1}^{K}}w_{k}I(j\in\ak)}.
\]

\subsection{Theoretical properties}

We will show consistency of the SOIL importance measure, under the condition
that the weighting vector $\mathbf{w}=(w_{1}\ddd w_{K})^{\intercal}$
satisfies the following properties referred to as \emph{weak consistency}
and \emph{consistency:}
\begin{defn}
[Weak Consistency and Consistency] The weighting vector $\mathbf{w}$
is \emph{weakly} \emph{consistent} if 
\begin{equation}
\dfrac{\sum_{k=1}^{K}w_{k}|\ak\nabla\as|}{r^{*}}\ \overset{p}{\to}\ 0,\qquad\mathrm{as}\ n\rightarrow\infty,\label{eq:weakly_consistent_def}
\end{equation}
and $\mathbf{w}$ is \emph{consistent} if
\[
\sum_{k=1}^{K}w_{k}|\ak\nabla\as|\ \overset{p}{\to}\ 0,\qquad\mathrm{as}\ n\rightarrow\infty,
\]
where $\nabla$ denotes the symmetric difference of two sets and $|\cdot|$ denotes number counting. 
\end{defn}
Intuitively, both weak consistency and consistency of weighting ensure that
the weighting of the candidate models is concentrated enough around
the true model to different degrees. Including the denominator $r^{*}$
in \eqref{eq:weakly_consistent_def} makes the weak consistency condition
more likely to be satisfied than consistency, when the true model
is allowed to increase in dimension as $n$ increases. There are several
different methods in the literature for providing the weight vector
$\mathbf{w}=(w_{1}\ddd w_{K})^{\intercal}$ for the candidate models
$\ba$. For example, \citet{buckland1997model} and \citet{leung2006information} studied a weighting
method based on information criterion, such as AIC \citep{AIC73}
and BIC \citep{schwarz1978estimating}; \citet{HMR1999} proposed the weighting by
Bayesian model averaging (BMA) from a Bayesian perspective; Several
attractive frequentist model averaging approaches are also developed (\citealp[e.g.][]{YY01,hjort2003frequentist,buckland1997model,hansen2007least,liang2012optimal,cheng2015toward,cheng2015forecasting}). In particular,
\citet{YY01} proposed a weighting strategy by data splitting and
cross-assessment, which is referred to as the adaptive regression
by mixing (ARM). He proved that the weighting by ARM delivers the
best rate of convergence for regression estimation. One advantage of ARM is that it can be applied to combine general regression procedures (not limited to parametric models). The ARM weighting
was extended to the classification problems in \citet{yang2000adaptive,yuan2008combining,zhang2013adaptively}.

Among the aforementioned weighting methods, there are several that give the consistent weights $\mathbf{w}$. For example, when there are a fixed number of models in the candidate model set,
BMA typically gives a consistent weighting. ARM also gives consistent
weighting when the data splitting ratio is properly chosen \citep{YY2007}. Now we prove that (a) under the assumption of weakly consistent weighting,
the sum of the SOIL importance of the true variables will tend to
the size of the true model $r^{*}$, while the sum of the SOIL importance
of the variables excluded by the true model converges to 0; (b)
a consistent weighting ensures that the SOIL importance of any true
variable tends to one as the sample size $n$ goes to infinity; while
each variable outside the true model will have the SOIL importance
tend to 0.
\begin{thm}
\label{thm:Under-the-assumption}(a) Under the assumption that the
weighting $\mathbf{w}$ is weakly consistent, we have:
\[
\dfrac{\sum_{j\in\mathcal{A}^{*}}S_{j}}{r^{*}}{\displaystyle \ \overset{p}{\to}\ }1,\qquad\dfrac{\sum_{j\notin\mathcal{A}^{*}}S_{j}}{r^{*}}{\displaystyle \ \overset{p}{\to}\ }0,\qquad\mathrm{as}\ n\rightarrow\infty;
\]
(b) When the weighting $\mathbf{w}$ is consistent, we have:
\[
\underset{j\in\as}{\min}\ S_{j}\overset{p}{\to}1,\qquad\underset{j\notin\as}{\max}\ S_{j}\overset{p}{\to}0,\qquad\mathrm{as}\ n\rightarrow\infty.
\]

\end{thm}
In some applications, one may set up a threshold value $c\in(0,1)$
for the variable importance, and only keeps all the variables whose
importances are greater than $c$. Denote by $\aa_{c}=\{j:S_{j}>c\}$
the model selected according to this criterion. The property of $\aa_{c}$ is
shown in the following theorem, which indicates that for any threshold
$c$, the number of the true variables missed by $\aa_{c}$ and the
number of the over-selected variables in $\aa_{c}$ will be relatively
small as $n$ grows large.
\begin{thm}
\label{thm:Let-,-}For any threshold $c\in(0,1)$, denote $\overline{\aa}_{c}=\{j\in\as:S_{j}\leq c,\ j=1,...,p\}$,
$\underline{\aa}_{c}=\{j\notin\aa^{*}:S_{j}>c,j=1,...,p\}$, then
if $\mathbf{w}$ is weakly consistent, we have

\[
\ \dfrac{|\overline{\aa}_{c}|}{r^{*}}{\displaystyle \ \overset{p}{\to}\ }0,\qquad\dfrac{|\underline{\aa}_{c}|}{r^{*}}{\displaystyle \ \overset{p}{\to}\ }0,\qquad\mathrm{as}\ n\rightarrow\infty.
\]

\end{thm}
As for the choice of threshold, its value depends on how one intends
to balance between the cost of overfitting and under-fitting. Actually
$|\aa_{c}\nabla\aa^{*}|=|\overline{\aa}_{c}\cup\underline{\aa}_{c}|$. We can
also get that $\dfrac{|\aa_{c}\nabla\aa^{*}|}{r^{*}}{\displaystyle \ \overset{p}{\to}\ }0$
as $n\rightarrow\infty$. The proofs of Theorem \ref{thm:Under-the-assumption} and Theorem
\ref{thm:Let-,-} are presented in the Appendix.

\section{IMPLEMENTATION}

\subsection{Candidate models}

Now we discuss how to choose candidate models for computing the SOIL
importance. One approach is to use a complete collection of all-subset
models as the candidate models, i.e. 
\[
\ba=\{\varnothing,\{j_{1}\}\ddd\{j_{p}\},\{j_{1},j_{2}\},\{j_{1,}j_{3}\}\ddd\{j_{1}\ddd j_{p}\}\},
\]
where $j_{1}\ddd j_{p}\in\{1\ddd p\}$. However, in the high-dimensional
setting when $p\gg n$, using the candidate models with all subsets
is computationally infeasible. Alternatively, we
obtain the candidate models using tools for high-dimensional penalized
regression
\begin{equation}
\min_{\boldsymbol{\beta}\in\mathbb{R}^{p}}\frac{1}{n}\sum_{i=1}^{n}(y_{i}-\mathbf{x}_{i}^{\intercal}\boldsymbol{\beta})^{2}+\sum_{j=1}^{p}p_{\lambda}(\beta_{j}),\label{eq:obj}
\end{equation}
where $p_{\lambda}(\cdot)$ is a nonnegative penalty function with
regularization parameter $\lambda\in(0,\infty)$, such as, Lasso \citep{tibshirani1996regression}
penalty $p_{\lambda}(u)=\lambda w|u|$ in \eqref{eq:obj}, and nonconvex
penalties including the smoothly clipped absolute deviation (SCAD)
penalty \citep{fan2001variable} 
\begin{eqnarray*}
p_{\lambda}(u) & = & \lambda|u|I(|u|\leq\lambda)+\left\{ \lambda|u|-\frac{(\lambda-|u|)^{2}}{2(\gamma-1)}\right\} I(\lambda<|u|\leq\gamma\lambda)\\
 &  & +\frac{(\gamma+1)\lambda^{2}}{2}I(|u|>\gamma\lambda),\qquad(\gamma>2),
\end{eqnarray*}
or the minimax concave penalty (MCP, \citealp{zhang2010nearly})
\[
p_{\lambda}(u)=\lambda\left(|u|-\frac{u^{2}}{2\gamma\lambda}\right)I(|u|\leq\gamma\lambda)+\frac{\gamma\lambda^{2}}{2}I(|u|>\gamma\lambda),\qquad(\gamma>1).
\]
We first apply a high-dimensional model selection method, e.g. SCAD,
on the data to compute solution paths for a sequence of tuning
parameter $\{\lambda_{1}\ddd\lambda_{L}\}$. Let $\{\widehat{\boldsymbol{\beta}}^{\lambda_{1}}\ddd\widehat{\boldsymbol{\beta}}^{\lambda_{L}}\}$
be the estimated coefficients of $L$ different regularization levels for the SCAD penalty and 
\[
\ba_{\mathrm{SCAD}}=\{\aa{}^{\lambda_{1}},\aa^{\lambda_{2}}\ddd\aa^{\lambda_{L}}\}
\]
be the resulting estimated models, where $\mathcal{A}^{\lambda_{l}}\equiv\mathrm{supp(\widehat{\boldsymbol{\beta}}^{\lambda_{l}})}=\{j:\widehat{\beta}_{j}^{\lambda_{l}}\neq0\}$.
We then use the set $\ba_{\mathrm{SCAD}}$ as the set of candidate
models. 

To further increase the chance of capturing the true/best model, we can
put together the resulting models from several different penalties to form
a larger set of candidate models, for example $\ba=\{\ba_{\mathrm{Lasso}},\ba_{\mathrm{AdaptiveLasso}},\ba_{\mathrm{SCAD}},\ba_{\mathrm{MCP}}\}$.
The individual penalized methods for producing $\ba$ do not have
to all contain the true model $\as$. As long as
there is at least one candidate model in the solution paths being (or very
close to) the true model, SOIL importance can still work well, provided that
the weighing is sensible.  By considering multiple model selection
methods through merging their solution paths, the chance of including
the true model in $\ba$ is enhanced.

\subsection{Weighting \label{sub:Estimating-the-weights}}

In this paper, we focus on two kinds of weighting methods: ARM weighting,
which is a weighting strategy by data splitting and cross-assessment,
and BIC weighing by BIC or a modified BIC information criterion (BIC-p) for
high dimensional data. \citet{yang1998asymptotic} also pointed out
that when we have exponentially many models, we should consider the
model complexity, which can also be interpreted as the prior probability
for the model. When the dimensionality is large, a uniform prior penalty
in ARM and BIC does not perform well. Following the same approach
in \citet{nanying}, we consider a non-uniform prior (or descriptive complexity from a coding perspective) $e^{-\psi C_{k}}$ 
when computing both then ARM weighting and the BIC weighting, where $\psi$ is a positive constant and
$C_k$ will be given in Algorithm 1.

\subsubsection*{Weighting using ARM with nonuniform priors.}

The ARM weighting method randomly splits the data $\mathbf{D}=\{(\mathbf{x}_{i},y_{i})\}_{i=1}^{n}$
into a training set $\mathbf{D}_{1}$ and a test set $\mathbf{D}_{2}$
of equal size (for simplicity, assume $n$ is an even number). Then
the regression models trained on $\mathbf{D}_{1}$ are
used for prediction on $\mathbf{D}_{2}$. Then the weights
$\mathbf{w}=(w_{1}\ddd w_{K})^{\intercal}$ can be computed based
on this prediction. Specifically, if we denote by $\boldsymbol{\beta}_{s}^{(k)}$
the nonzero-coefficient sub-vector of $\boldsymbol{\beta}^{(k)}$
specified by the model $\ak$, and let $\mathbf{x}_{s}^{(k)}\in\mathbb{R}^{|\ak|}$
be the corresponding subset of predictors, we summarize the ARM weighting
method in Algorithm \ref{alg:ARM-weighting-method}.

\begin{algorithm}[H]
\begin{itemize}[nolistsep]
\item {\footnotesize{}Randomly split $\mathbf{D}$ into a training set $\mathbf{D}_{1}$
and a test set $\mathbf{D}_{2}$ of equal size. }{\footnotesize \par}
\item {\footnotesize{}For each $\ak\in\ba$, fit a standard linear regression
of $y$ on $\mathbf{x}_{s}^{(k)}$ using the training set $\mathbf{D}_{1}$
and get the estimated $\widehat{\boldsymbol{\beta}}_{s}^{(k)}$ and
$\widehat{\boldsymbol{\sigma}}_{s}^{(k)}$. }{\footnotesize \par}
\item {\footnotesize{}For each $\ak$, compute the prediction $\mathbf{x}_{s}^{(k)\intercal}\widehat{\boldsymbol{\beta}}_{s}^{(k)}$
on the test set $\mathbf{D}_{2}$. }{\footnotesize \par}
\item {\footnotesize{}Compute the weight $w_{k}$ for each candidate model:
\[
w_{k}=\frac{e^{-\psi C_{k}}(\widehat{\boldsymbol{\sigma}}_{s}^{(k)})^{-n/2}\prod_{i\in\mathbf{D}_{2}}\exp(-(\widehat{\boldsymbol{\sigma}}_{s}^{(k)})^{-2}(y_{i}-\mathbf{x}_{s,i}^{(k)\intercal}\widehat{\boldsymbol{\beta}}_{s}^{(k)})^{2}/2)}{\sum_{l=1}^{K}e^{-\psi C_{l}}(\widehat{\boldsymbol{\sigma}}_{s}^{(l)})^{-n/2}\prod_{i\in\mathbf{D}_{2}}\exp(-(\widehat{\boldsymbol{\sigma}}_{s}^{(l)})^{-2}(y_{i}-\mathbf{x}_{s,i}^{(l)\intercal}\widehat{\boldsymbol{\beta}}_{s}^{(k)})^{2}/2)},
\]
for $k=1\ddd K$, where $C_{k}=s_{k}\log{\frac{e\cdot p}{s_{k}}}+2\log(s_{k}+2)$.}{\footnotesize \par}
\item {\footnotesize{}Repeat the steps above (with random data splitting)
$L$ times to get $w_{k}^{(l)}$ for $l=1\ddd L$, and get $w_{k}=\frac{1}{L}\sum_{l=1}^{L}w_{k}^{(l)}$. }{\footnotesize \par}
\end{itemize}
\caption{{\footnotesize The procedure of the ARM weighting for the regression case.}\label{alg:ARM-weighting-method}}
\end{algorithm}

\subsubsection*{Weighting using information criteria with nonuniform priors.}

An alternative way of weighting is using BIC information criteria. Define
$I_{k}^{\mathrm{BIC}}=-2(\log\ell_{k}+s_{k}\log n)$ as the BIC information
criterion, where $\ell_{k}$ is the maximized likelihood for model
$k$ and $s_{k}$ denotes the number of non-constant predictors. Then
weight $w_{k}$ for model $\ak\in\ba$ is computed by 
\begin{equation}
w_{k}=\exp(-\frac{I_{k}}{2}-\psi C_{k})/\sum_{l=1}^{K}\exp(-\frac{I_{l}}{2}-\psi C_{l}).\label{eq:bic}
\end{equation}
We refer to the above approach with nonuniform priors as the BIC-p weighting.

Besides the ARM and BIC-p weighting, one can also consider another alternative weighting approach by using Fisher's fiducial idea from the generalized fiducial inference \citep{lai2015generalized}. The details are included in Supplementary Materials Part A. We do not discuss this method in details since it only applies to the regression settings.

Often consistency of a weighting method is proved when all subset
models are considered (\citealp[e.g.][]{lai2015generalized}). But
when $p$ is large, it is computationally infeasible to include all the
variables, so we need some screening methods to reduce the number
of variables. Next we prove the consistency of SOIL importance: 
\begin{defn}
[Path-consistent] A method is called path-consistent if 
\[
P(\as\in\Delta)\rightarrow1,\ as\ n\rightarrow\infty,
\]
where $\Delta$ denotes the whole solution paths produced by the method.\end{defn}
\begin{cor}
\label{corollary}Under the assumption that the weighting $\mathbf{w}$ on the all-subset candidate models $\ba$ is consistent, as long
as at least one method is path-consistent, we have 
\[
\underset{j\in\as}{\min}\ S(j;\mathbf{w^{'}},\ba^{'})\overset{p}{\to}1,\qquad\underset{j\notin\as}{\max}\ S(j;\mathbf{w^{'}},\ba^{'})\overset{p}{\to}0,\qquad\mathrm{as}\ n\rightarrow\infty,
\]
where $\mathbf{w^{'}}$ is the renormalized weighting on $\ba^{'}$,
which is the collection of models using union of solution paths.
\end{cor}

\subsection{Software}

We provide our implementation of the SOIL importance measure in an official \textbf{R} package \texttt{SOIL}, which is publicly available from the Comprehensive R Archive Network at \url{https://cran.r-project.org/web/packages/SOIL/index.html}. The package is also provided in the supplementary materials.

\section{EXTENSION TO THE BINARY CLASSIFICATION MODEL \label{sub:Extended-to-the}}

We extend the SOIL importance to the binary logistic regression case. Let $Y\in\{0,1\}$ be the response variable 
and $X\in\mathbb{R}^{p}$ be the predictor vector. We assume that $Y$
has a Bernoulli distribution with conditional probabilities
\begin{equation}
\mathrm{Pr}(Y=1|X=\mathbf{x})  = 1-\mathrm{Pr}(Y=0|X=\mathbf{x})=\frac{e^{\mathbf{x}^{\intercal}\boldsymbol{\beta^{*}}}}{1+e^{\mathbf{x}^{\intercal}\boldsymbol{\beta^{*}}}},
\end{equation}
where $\boldsymbol{\beta}^{*}=(\beta_{1}^{*}\ddd\beta_{p}^{*})^{\intercal}$
is the vector corresponding to the true underlying
model.  
The ARM weighting for the logistic regression can be computed by
Algorithm \ref{alg:ARM-weighting-method-1}.

\begin{algorithm}[H]
\begin{itemize}[nolistsep]
\item {\footnotesize{}Randomly split $\mathbf{D}$ into a training set $\mathbf{D}_{1}$
and a test set $\mathbf{D}_{2}$ of equal size. }{\footnotesize \par}
\item {\footnotesize{}For each $\ak\in\ba$, fit a standard logistic regression
of $y$ on $\mathbf{x}_{s}^{(k)}$ using the samples in $\mathbf{D}_{1}$
and get the fitted value $\widehat{p}^{(k)}(\mathbf{x}_{s}^{(k)})$,
\begin{align*}
\widehat{p}^{(k)}(\mathbf{x}_{s}^{(k)}) & \equiv\mbox{Pr}(Y=1|X_{s}^{(k)}=\mathbf{x}_{s}^{(k)})\\
 & =\exp(\mathbf{x}_{s}^{(k)\intercal}\widehat{\boldsymbol{\beta}}_{s}^{(k)})/(1+\exp(\mathbf{x}_{s}^{(k)\intercal}\widehat{\boldsymbol{\beta}}_{s}^{(k)})),\qquad k=1\ddd K.
\end{align*}
}{\footnotesize \par}
\item {\footnotesize{}For each $\ak$, compute the prediction $\widehat{p}^{(k)}(\mathbf{x}_{s}^{(k)})$
on the test set $\mathbf{D}_{2}$. }{\footnotesize \par}
\item {\footnotesize{}Compute the weight $w_{k}$ for each candidate model:
\[
w_{k}=\frac{e^{-\psi C_{k}}\prod_{i\in\mathbf{D}_{2}}\widehat{p}^{(k)}(\mathbf{x}_{s,i}^{(k)})^{y_{i}}\left(1-\widehat{p}^{(k)}(\mathbf{x}_{s,i}^{(k)})\right)^{1-y_{i}}}{\sum_{l=1}^{K}e^{-\psi C_{l}}\prod_{i\in\mathbf{D}_{2}}\widehat{p}^{(l)}(\mathbf{x}_{s,i}^{(l)})^{y_{i}}\left(1-\widehat{p}^{(l)}(\mathbf{x}_{s,i}^{(l)})\right)^{1-y_{i}}},
\]
for $k=1\ddd K$, where $C_{k}=s_{k}\log{\frac{e\cdot p}{s_{k}}}+2\log(s_{k}+2)$.}{\footnotesize \par}
\item {\footnotesize{}Repeat the steps above (with random data splitting)
$L$ times to get $w_{k}^{(l)}$ for $l=1\ddd L$, and get $w_{k}=\frac{1}{L}\sum_{l=1}^{L}w_{k}^{(l)}$. }{\footnotesize \par}
\end{itemize}
\caption{{\footnotesize The procedure of the ARM weighting for the binary classification case.}\label{alg:ARM-weighting-method-1}}
\end{algorithm}

\subsection{Weighting using information criteria with nonuniform priors}

Similarly, the weight $w_{k}$ for model $\ak\in\ba$ using BIC-p the information
criterion can be computed in the same way as in \eqref{eq:bic}
where $I_{k}^{\mathrm{BIC}}=-2\log\ell_{k}+2s_{k}\log n$, with \textbf{$s_{k}=|\ak|$}
and $\ell_{k}$ being the maximized likelihood function for the logistic
model $\ak$.


\section{SIMULATIONS}

In this section, we consider a number of simulation settings to highlight the
properties of SOIL in contrast to some other importance measures.
We compare SOIL using the ARM and BIC-p weighting with three variable importance alternatives, which are denoted as LMG,
RFI1 and RFI2. LMG is the relative importance measure by averaging
over all possible orderings for $R^{2}$ decomposition \citep{lindeman1980introduction}.
RFI1 and RFI2 are importance measures in random forests proposed by \citet{breiman2001random}.
Specifically, RFI1 is computed from a normalized difference between
the prediction error on the out-of-bag (OOB) portion of the data and
that on the permuted OOB data for each predictor variable. RFI2 is
the total decrease in node impurities from splitting on a particular variable,
averaged over all trees. The node impurity is defined by the Gini
index for classification, and by residual sum of squares
for regression. Computationally, LMG can be obtained by the R implementation
\texttt{relaimpo} \citep{gromping2006relative}, while RFI1 and RFI2
can be obtained by R implementation \texttt{randomForest} \citep{Liaw2002}.
Since LMG can only handle the linear case with up to about 20 variables
due to its computational limitation, we are not able to get the relative
importance LMG in some of our examples. The choice of the prior $\psi$ for the ARM
and BIC-p weighting can be specified by the users. To avoid
cherry-picking, we present the results with a fixed choice: $\psi=0.5$.
Our experience is that $\psi=0.5$ or $1$ generally works quite well.

In the following we compare different variable importance measures
for Gaussian and Binomial cases under various settings of sample sizes,
dimensions and feature correlations.

\paragraph{Model 1: Gaussian.}

The simulation data $\{y_{i},\mathbf{x}_{i}\}_{i=1}^{n}$ is generated
from the linear model $y_{i}=\mathbf{x}_{i}^{\intercal}\boldsymbol{\beta}^{*}+\epsilon_{i}$,
$\epsilon_{i}\sim N(0,\sigma^{2})$ and $\sigma\in\{0.1,5\}$. We generate
$\mathbf{x}_{i}$ from multivariate normal distribution $N_{p}(0,\Sigma)$.
For each element $\Sigma_{ij}$ of $\Sigma$, $\Sigma_{ij}=\rho^{\left|i-j\right|}$,
i.e. the correlation of $X_{i}$ and $X_{j}$ is $\rho^{\left|i-j\right|}$,
with $\rho\in\{0,0.9\}$.

\paragraph{Model 2: Binomial.}

The i.i.d. sample $\{y_{i},\mathbf{x}_{i}\}_{i=1}^{n}$ is generated
from the binomial model $\mathrm{logit}(p_{i})=\mathbf{x}_{i}^{\intercal}\boldsymbol{\beta^{*}}$,
where $p_{i}=P(Y=1|X=\mathbf{x}_{i})$. And $\mathbf{x}_{i}$ is generated
in the same way as the Gaussian case.

We summarize in Table \ref{tab:simu settings} the model settings
adopted in this simulation. For each model setting with a specific
choice of the parameters $(\rho,\sigma^{2})$, we repeat the simulation
100 times and compute the averaged variable importance measures for
SOIL-BIC-p, SOIL-ARM, LMG, RFI1 and RFI2. Due to page restrictions, the figures of Example \ref{fig:resS1}--\ref{fig:resS5} as defined in Table \ref{tab:simu settings} are only provided in the supplementary materials, while the summary of all the examples are discussed in the main part of the paper. 

The results for the simulations are shown in Figure
\ref{fig:res1}--\ref{fig:res6} and Figure \ref{fig:resS1}--\ref{fig:resS5} in the Supplementary Materials Part B. For the scaling of the importance measures, we standardize RFI1 and RFI2, dividing them by their respective maximum value of the variable importance among all the variables for each realization of the data. As a result, in each figure, we can see that the maximum value of RFI1 or RFI2 (after the standardization) is always one. For SOIL and LMG, we keep their original values as being proposed. The fact that the LMG importance values sum to one over the variables should be kept in mind when comparing the different importance measures on the graphs.  

\begin{table}[H]
\ra{0.8} 
\begin{centering}
{\footnotesize{}}%
\begin{tabular}{llll}
\toprule 
{\footnotesize{}Example } & {\footnotesize{}$n$ } & {\footnotesize{}$p$ } & {\footnotesize{}Model Settings}\tabularnewline
\midrule 
\multicolumn{4}{c}{{\footnotesize{}Gaussian Case}}\tabularnewline
{\footnotesize{}1 } & {\footnotesize{}100 } & {\footnotesize{}200 } & {\footnotesize{}$\boldsymbol{\beta}^{*}=(4,4,4,-6\sqrt{2},\frac{3}{4},0,...,0)^{\intercal}$}\tabularnewline
{\footnotesize{}2 } & {\footnotesize{}150 } & {\footnotesize{}14+1 } & {\footnotesize{}}%
\parbox[t]{11cm}{%
{\footnotesize{}$\boldsymbol{\beta}^{*}=(4,4,4,-6\sqrt{2},\frac{3}{4},0,...,0)^{\intercal}$.
Add $X_{15}=0.5X_{1}+2X_{4}+e$ and $\beta_{15}^{*}=0$, where $e\sim N(0,0.01)$.}%
}\tabularnewline

{\footnotesize{}3 } & {\footnotesize{}150 } & {\footnotesize{}8 } & {\footnotesize{}$\boldsymbol{\beta}^{*}=(0\ddd0){}^{\intercal}$}\tabularnewline
{\footnotesize{}4 } & {\footnotesize{}150 } & {\footnotesize{}8 } & {\footnotesize{}$\boldsymbol{\beta}^{*}=(1\ddd1){}^{\intercal}$}\tabularnewline
{\footnotesize{}S1 } & {\footnotesize{}150 } & {\footnotesize{}20 } & {\footnotesize{}$\boldsymbol{\beta}^{*}=(4,4,4,-6\sqrt{2},\frac{3}{4},0,...,0)^{\intercal}$}\tabularnewline
{\footnotesize{}S2 } & {\footnotesize{}150 } & {\footnotesize{}6+6 } & {\footnotesize{}}%
\parbox[t]{11cm}{%
{\footnotesize{}$\boldsymbol{\beta}^{*}=(4,4,-6\sqrt{2},\frac{3}{4},0,0)^{\intercal}$.
Add $(X_{1}^{2},X_{2}^{2},X_{3}^{2},X_{4}^{2},X_{5}^{2},X_{6}^{2})$
and corresponding coefficients $(\beta_{7}^{*},\beta_{8}^{*}\ddd\beta_{12}^{*})^{\intercal}=(4,0,1,0,0,0)^{\intercal}$.}%
}\tabularnewline
{\footnotesize{}S3 } & {\footnotesize{}150 } & {\footnotesize{}6+6 } & {\footnotesize{}}%
\parbox[t]{11cm}{%
{\footnotesize{}$\boldsymbol{\beta}^{*}=(4,4,-6\sqrt{2},\frac{3}{4},0,0)^{\intercal}$.
}\\
{\footnotesize{}Add $(X_{1}X_{2},X_{1}X_{3},X_{1}X_{4},X_{2}X_{3},X_{2}X_{4},X_{3}X_{4})$
and corresponding coefficients $(\beta_{7}^{*},\beta_{8}^{*}\ddd\beta_{12}^{*})^{\intercal}=(4,2,2,0,0,0)^{\intercal}$.}%
}\tabularnewline
\midrule 
\multicolumn{4}{c}{{\footnotesize{}Binomial Case}}\tabularnewline
{\footnotesize{}5 } & {\footnotesize{}80} & {\footnotesize{}6 } & {\footnotesize{}$\boldsymbol{\beta}^{*}=\left(1,\frac{1}{2},\frac{1}{3},\frac{1}{4},\frac{1}{5},\frac{1}{6},0\right){}^{\intercal}$}\tabularnewline
{\footnotesize{}6 } & {\footnotesize{}5000 } & {\footnotesize{}6 } & {\footnotesize{}$\boldsymbol{\beta}^{*}=\left(1,\frac{1}{2},\frac{1}{3},\frac{1}{4},\frac{1}{5},\frac{1}{6},0\right){}^{\intercal}$}\tabularnewline
{\footnotesize{}S4} & {\footnotesize{}150 } & {\footnotesize{}20 } & {\footnotesize{}$\boldsymbol{\beta}^{*}=(4,4,4,-6\sqrt{2},\frac{3}{4},0,...,0){}^{\intercal}$}\tabularnewline
{\footnotesize{}S5} & {\footnotesize{}100 } & {\footnotesize{}200 } & {\footnotesize{}$\boldsymbol{\beta}^{*}=(4,4,4,-6\sqrt{2},\frac{3}{4},0,...,0){}^{\intercal}$}\tabularnewline
\bottomrule
\end{tabular}
\par\end{centering}{\footnotesize \par}
\caption{Simulation settings\label{tab:simu settings}}
\end{table}

\subsection{Relative performances of importance measures in several key aspects\label{sec:property}}

A summary of the relevant properties of different important measures
is provided in Table \ref{tab:Comparison-of-the}.  In the following
we discuss point-by-point these characteristics for the importance
measures in comparison. For convenience, we call the variables with nonzero coefficients the ``true'' variables.

\begin{table}
\ra{0.8}
\begin{centering}
{\footnotesize{}}%
\begin{tabular}{lccccc}
\toprule 
 & {\footnotesize{}SOIL-ARM } & {\footnotesize{}SOIL-BIC-p } & {\footnotesize{}LMG } & {\footnotesize{}RFI1 } & {\footnotesize{}RFI2}\tabularnewline
\midrule 
{\footnotesize{}Inclusion/Exclusion } & {\footnotesize{}$\checkmark$ } & {\footnotesize{}$\checkmark$ } &  &  & \tabularnewline
{\footnotesize{}Tuning in to information } & {\footnotesize{}$\checkmark$ } & {\footnotesize{}$\checkmark$ } &  &  & \tabularnewline
{\footnotesize{}Robustness to feature correlation } & {\footnotesize{}$\checkmark$ } & {\footnotesize{}$\checkmark$ } &  &  & \tabularnewline
{\footnotesize{}Robustness against confuser } & {\footnotesize{}$\checkmark$ } & {\footnotesize{}$\checkmark$ } &  &  & \tabularnewline
{\footnotesize{}Sensitivity to high-order terms } & {\footnotesize{}$\checkmark$ } & {\footnotesize{}$\checkmark$ } &  &  & \tabularnewline
{\footnotesize{}Pure relativeness } &  &  & {\footnotesize{}$\checkmark$ } & {\footnotesize{}$\checkmark$ } & {\footnotesize{}$\checkmark$}\tabularnewline
{\footnotesize{}Order preserving } & {\footnotesize{}$\checkmark$ } & {\footnotesize{}$\checkmark$ } &  &  & \tabularnewline
{\footnotesize{}High-dimensionality } & {\footnotesize{}$\checkmark$ } & {\footnotesize{}$\checkmark$ } &  & {\footnotesize{}$\checkmark$ } & {\footnotesize{}$\checkmark$}\tabularnewline
{\footnotesize{}Non-parametricness } &  &  &  & {\footnotesize{}$\checkmark$ } & {\footnotesize{}$\checkmark$}\tabularnewline
{\footnotesize{}Non-negativity } & {\footnotesize{}$\checkmark$ } & {\footnotesize{}$\checkmark$ } & {\footnotesize{}$\checkmark$ } &  & {\footnotesize{}\checkmark}\tabularnewline
\bottomrule
\end{tabular}
\par\end{centering}{\footnotesize \par}
\caption{Comparison of the characteristics for the importance measures. A ``$\checkmark$\textquotedblright{}
indicates that a specified method has the given property. A blank
space indicates the absence of a property. \label{tab:Comparison-of-the}}
\end{table}

\paragraph{Inclusion/exclusion.}

The inclusion/exclusion aspect addresses the issue if an importance
measure can give a proper sense if a predictor is likely to be needed
in the best model to describe the data. These two criteria for importance have been discussed in \citet{gromping2015variable}. Recall that given enough data
for SOIL importance, the true variables in the model have large importances
(inclusion) and the variables that are not in the true model have
importances around zero (exclusion). In all examples, we can see that
the SOIL-BIC-p and SOIL-ARM have inclusion/exclusion properties. For
example in Figure \ref{fig:resS1}, all the true variables $(X_{1}\ddd X_{5})$
have their SOIL importances around one, even though their coefficients are
different, i.e. $(\beta_{1}^{*}\ddd\beta_{5}^{*})=(4,4,4,-6\sqrt{2},\frac{3}{4})$.
In contrast, the other three measures LMG, RFI1 and RFI2 do not have
the inclusion property when $\rho=0$ and $\sigma^{2}=0.01$ (they all undervalue the importance of $X_5$, which has a small coefficient). LMG, RFI1 and RFI2 do not have the exclusion property either. We can see that in Figure \ref{fig:res2} the noise variable $X_{15}$ confuses LMG, RF1 and RF2. In Figure \ref{fig:resS2} when $\rho=0.9$, LMG, RF1 and RF2 assign relatively high values on the noise variable $X_{8}$. In Figure \ref{fig:resS3} when $\rho=0.9$ and $\sigma^{2}=25$, LMG, RF1 and RF2 fail on the noise variable $X_{10}$

SOIL is certainly incapable of giving high importance to very weak variables in the true model.
For example Figure \ref{fig:res5} shows that in a binomial model with the decreasing
coefficient vector $\boldsymbol{\beta}^{*}=\left(1,\frac{1}{2},\frac{1}{3},\frac{1}{4},\frac{1}{5},\frac{1}{6},0\right){}^{\intercal}$,
the true variable $X_{6}$'s SOIL importance is only around 0.1, not much above that of the noise variable $X_{7})$. However this problem is alleviated
as the sample increases: Figure \ref{fig:res6} shows that the SOIL-ARM
and SOIL-BIC-p importances of six true variables $(X_{1}\ddd X_{6})$
become closer to one when $n$ increases from 80 to 5000. In contrast,
the LMG, RFI1, and RFI2 stay basically the same as the sample size
increases.

\paragraph{Tuning in to information.}

For high dimensional data, more often than not (to say the
least), sparsity is a reluctant acceptance that the info and/or computational
limit only allows us a simple model for application. The optimal sparsity
should depend on the sample size and noise level. Therefore, it is
desirable to have an importance measure to honor this perspective.
When the sample size increases or the noise decreases, we should have
more information. Thus, the importance obtained from the data
should change due to the enrichment of information. Therefore in most examples,
when the correlation $\rho$ and $\sigma^{2}$ are low, one may hope the variable importances
delineate the true model. Comparing Examples \ref{fig:res5} and \ref{fig:res6}, which differ only in the sample size, as shown in Figure \ref{fig:res5} and Figure \ref{fig:res6}, only SOIL-BIC-p and SOIL-ARM react to the much increased information due to sample size increase, while the other three
importances are not tuned in to the information change.

\paragraph{Robustness to feature correlation.}

SOIL importances show robustness against noise increase and higher feature correlation.
For example in Figure \ref{fig:res1}, \ref{fig:res2} and Figure \ref{fig:resS1}--\ref{fig:resS5} in  Supplementary Materials Part B, even when there is high feature correlation $(\rho=0.9,\sigma^{2}=0.01)$
or strong noise $(\rho=0,\sigma^{2}=25)$ in the data, the SOIL-BIC-p
and SOIL-ARM can still select the true variable $X_{5}$ while the
other methods consider $X_{5}$ as unimportant. But in a case of both high feature correlation and strong noise $(\rho=0.9,\sigma^{2}=25)$,
none of the importance measures in comparison can quite clearly select $X_{5}$
as an important variable because the information is too limited.

\paragraph{Robustness against confusers.}

A confuser refers to a variable that is closely related to a true variable or some linear combination
of the true variables but not to the extent of serving as a valid alternative. An importance measure oriented towards sparse modeling 
should assign near zero importances on the confusers. The simulation results show that the SOIL importance measures
are much more robust to confusers than LMG, RFI1 and RFI2. In Example \ref{fig:res2},
we generate a confuser $X_{15}=0.5X_{1}+2X_{4}+e$ with Gaussian
noise $e\sim N(0,0.01)$. The results in Figure \ref{fig:res2} show
that LMG, RFI1 and RFI2 fail to assign small importance to $X_{15}$ (not in the true model) and view it more important than some true variables. In contrast, small ARM and BIC-p importances
for $X_{15}$ correctly indicate that it is unimportant.

\paragraph{Sensitivity to higher-order terms.}

The SOIL importance measures are more sensitive to inclusion of higher-order terms
in the model. In Example \ref{fig:resS2} and \ref{fig:resS3} we add quadratic terms
$X_{1}^{2}$, $X_{2}^{2}$, $X_{3}^{2}$, $X_{4}^{2}$, $X_{5}^{2}$, $X_{6}^{2}$ and pairwise interactions $X_{1}X_{2}$, $X_{1}X_{3}$, $X_{1}X_{4}$, $X_{2}X_{3}$, $X_{2}X_{4}$, $X_{3}X_{4}$
respectively, where the coefficients for $X_{1}X_{2}$, $X_{1}X_{3}$, $X_{1}X_{4}$
and $X_{1}^{2}$, $X_{3}^{2}$ are nonzero in the true models. Results
in Figure \ref{fig:resS2} and \ref{fig:resS3} show that the ARM and
BIC-p methods can select both true main-effect variables and true higher-order
terms, whereas LMG, RFI1 and RFI2 fail to select some of the main-effect
variables when interactions or quadratic terms are included.

\paragraph{Pure relativity.}

An importance measure is said to be purely relative if the values
individually do not have a sensible meaning on their own. One drawback
of an importance measure with pure relativity is that it does not
differentiate between equal importance and equal unimportance cases.
All coefficients in Example \ref{fig:res3} and \ref{fig:res4} have the same relative size, which
are $\boldsymbol{\beta}^{*}=(0\ddd0){}^{\intercal}$ and $\boldsymbol{\beta}^{*}=(1\ddd1){}^{\intercal}$
respectively. We find that LMG, RFI1 and RFI2 do not offer any clue on importance of each variable itself. Variables $(X_{1}\ddd X_{6})$
in Example \ref{fig:res3} have very similar LMG and RFI2 values to those in Example
\ref{fig:res4}. And RFI1 behaves wildly as it assigns very much different importances to the variables in the independence case ($\rho=0$) of Example \ref{fig:res3}. The importance values are even significantly negative for some variables. In contrast, SOIL-BIC-p and SOIL-ARM nicely separate the two examples. 

\paragraph*{Order preserving.}

Order preserving refers to the property that the importance 
reflects the ``order" of the variables or not: (1) For the true variables (standardized) with not too high correlations with others,
it may be natural to expect the ones with larger coefficients to have larger importances (up to one of course); (2)
The true variables should have larger importances compared to the noise
ones. In the case that the sample size is too small for some true variables to be detectable, the order preserving property demands that the noise variables should not receive significantly higher importance values than these subtle true variables. SOIL-BIC-p and SOIL-ARM exhibit the order preserving property
in all the cases. LMG behaves poorly when there exists a confuser
as in Figure \ref{fig:res2}. RFI1 and RFI2 do not preserve the order
when correlation $\rho=0.9$ and/or noise $\sigma^{2}$ is large.

\paragraph{High-dimensionality.}

SOIL-BIC-p, SOIL-ARM, RFI1 and RFI2 can work for high-dimensional
data when $p>n$ as shown in Figure \ref{fig:res1} and \ref{fig:resS5}. The exclusion
and inclusion properties still hold for SOIL-BIC-p and SOIL-ARM in the
high dimensional case (inclusion of a weak variable requires that $\sigma^{2}$ is not too high). In contrast, LMG does not support high-dimensional
data.

\paragraph{Non-negativity.}

SOIL-BIC-p, SOIL-ARM, LMG and IMG2 always yield non-negative importance value. However, RFI1 does not satisfy this criterion.

\paragraph{Non-parametricness.}

Among the importance measures, only the two from random forest are not limited to parametric modeling.

\begin{figure}
\begin{centering}
\includegraphics[scale=0.64]{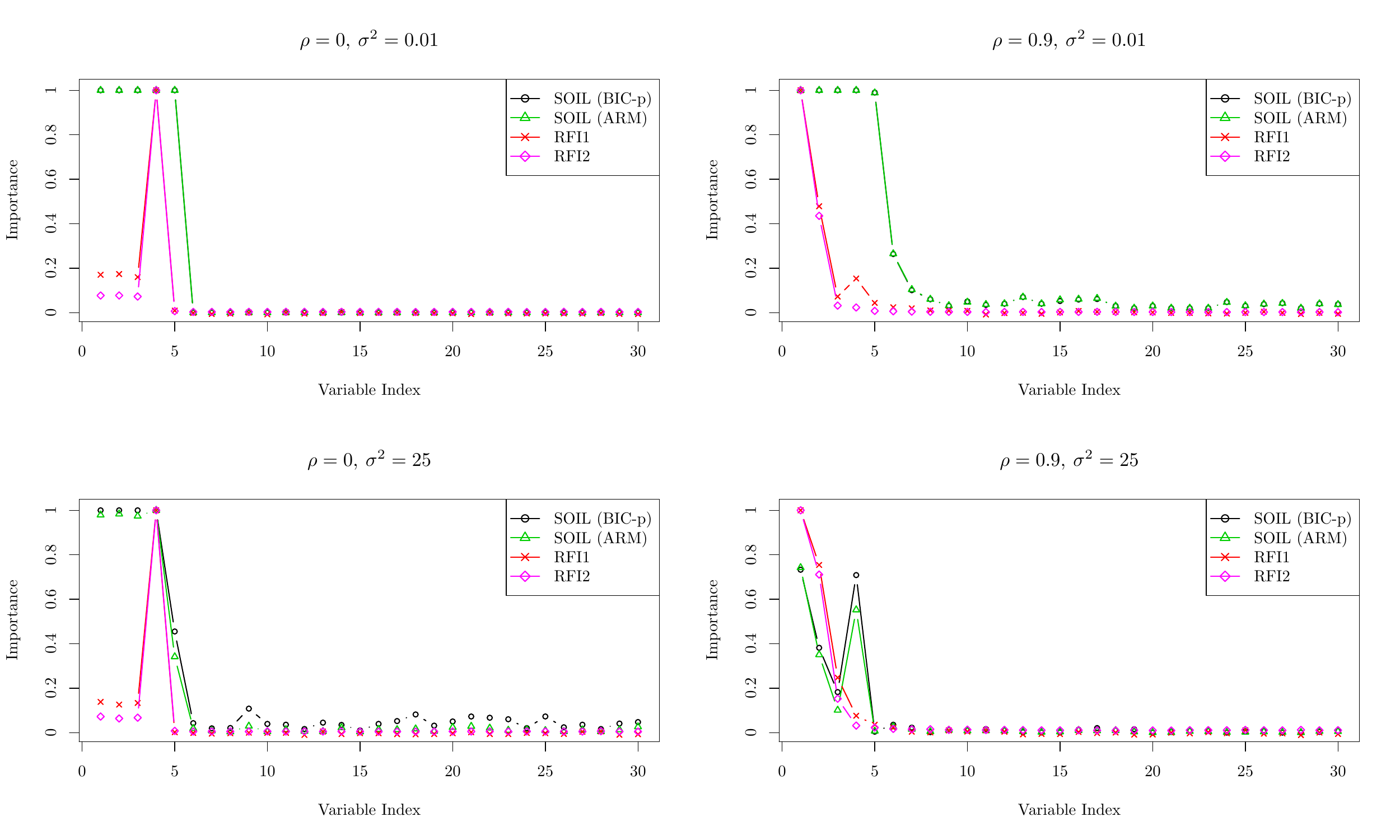}
\par\end{centering}

\caption{Simulation results for Example \ref{fig:res1}, where $n=100$, $p=200$. The true
coefficients $\boldsymbol{\beta}^{*}=(4,4,4,-6\sqrt{2},\frac{3}{4},0,...,0)$.\label{fig:res1}}
\end{figure}

\begin{figure}
\begin{centering}
\includegraphics[scale=0.6]{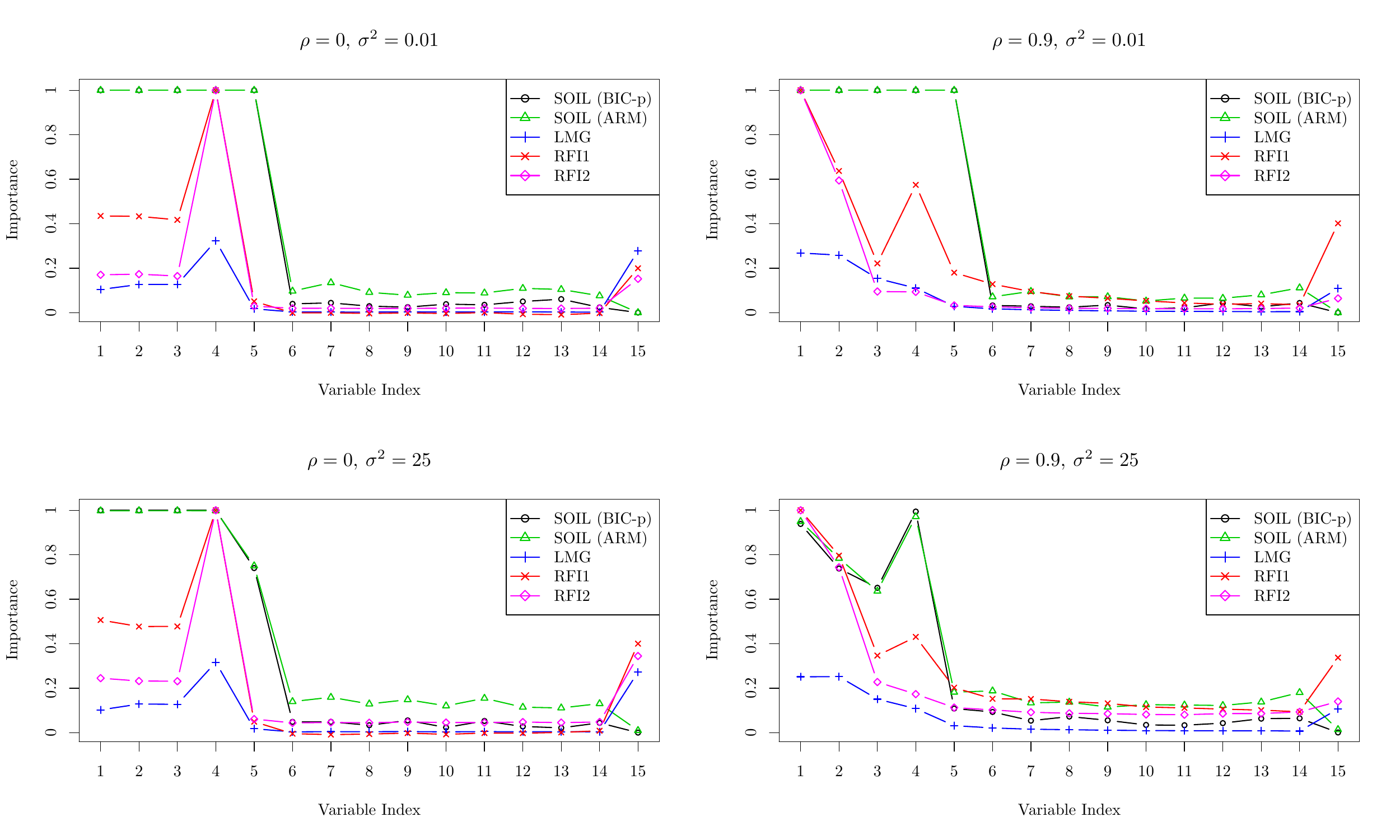}
\par\end{centering}

\caption{Simulation results for Example \ref{fig:res2}, where $n=150$, $p=14$. The true
coefficients $\boldsymbol{\beta}^{*}=(4,4,4,-6\sqrt{2},\frac{3}{4},0,...,0)$.
Add $X_{15}=0.5*X_{1}+2*X_{4}+e$ and corresponding $\beta_{15}^{*}=0$,
where $e\sim N(0,\sigma_{e}^{2})$.\label{fig:res2}}
\end{figure}

\begin{figure}
\begin{centering}
\includegraphics[scale=0.64]{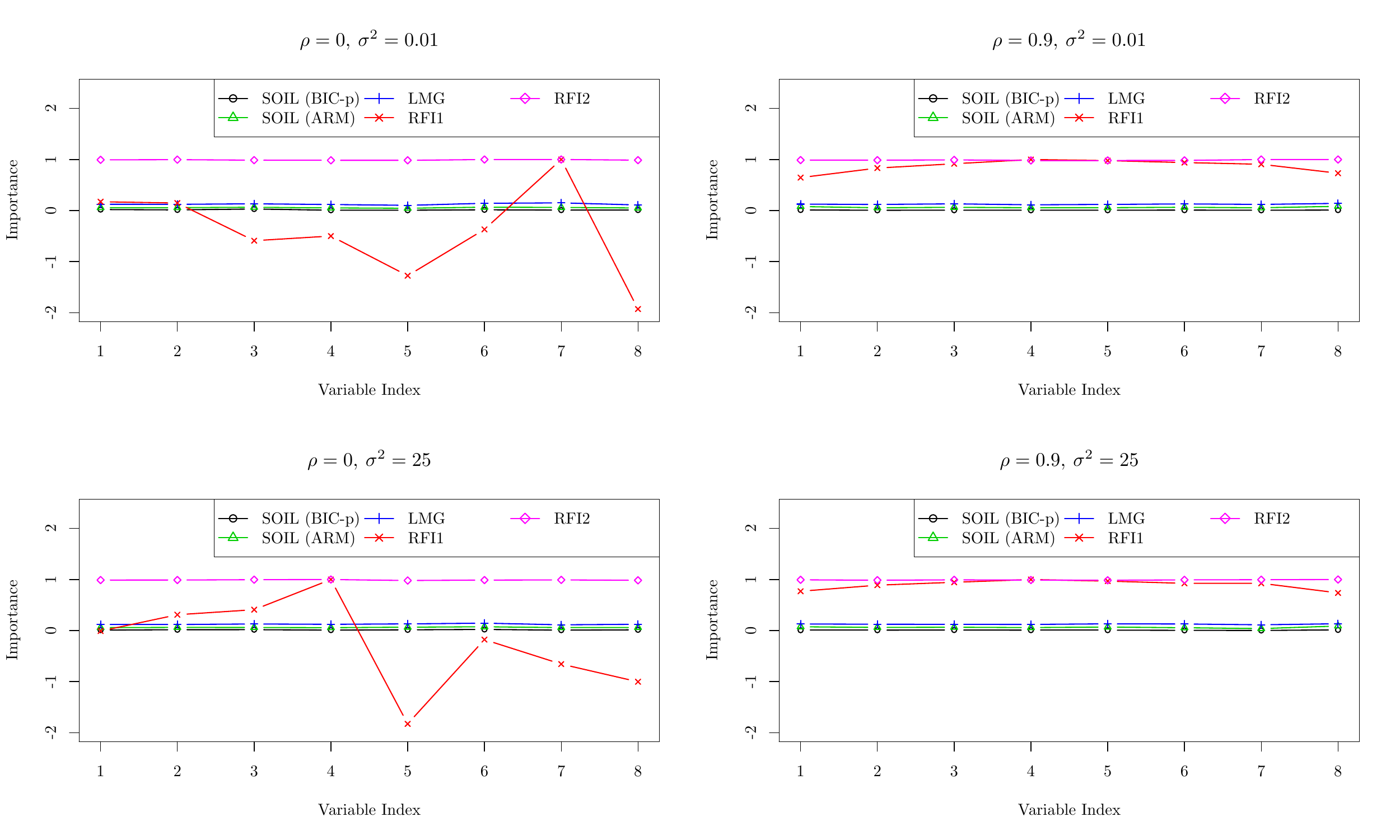}
\par\end{centering}

\caption{Simulation results for Example \ref{fig:res3}, where $n=150$, $p=8$. The true
coefficients $\boldsymbol{\beta}^{*}=(0\protect\ddd0)^{\intercal}$.\label{fig:res3}}
\end{figure}

\begin{figure}
\begin{centering}
\includegraphics[scale=0.64]{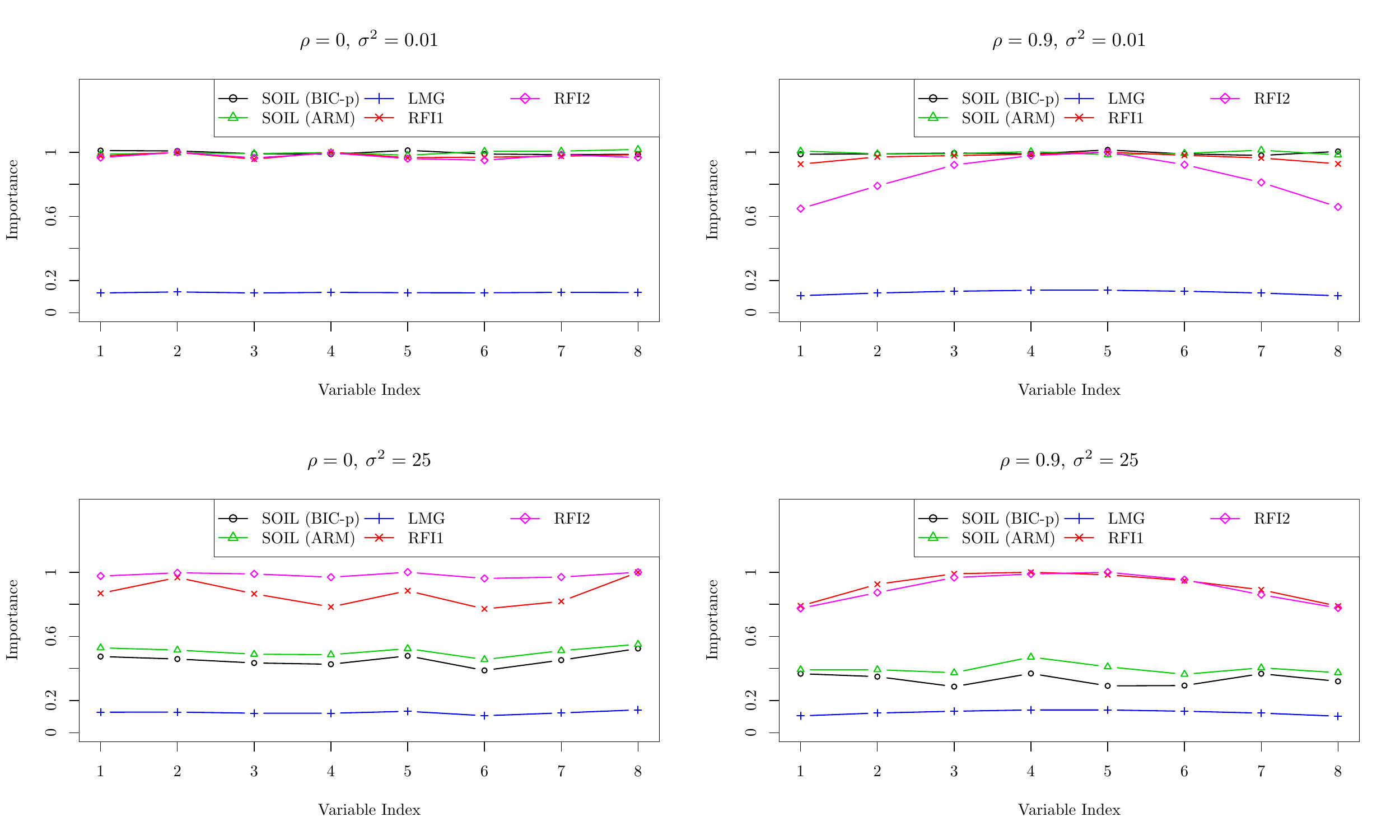}
\par\end{centering}

\caption{Simulation results for Example \ref{fig:res4}, where $n=150$, $p=8$. The true
coefficients $\boldsymbol{\beta}^{*}=(1\protect\ddd1)^{\intercal}$.\label{fig:res4}}
\end{figure}

\begin{figure}
\begin{centering}
\includegraphics[scale=0.64]{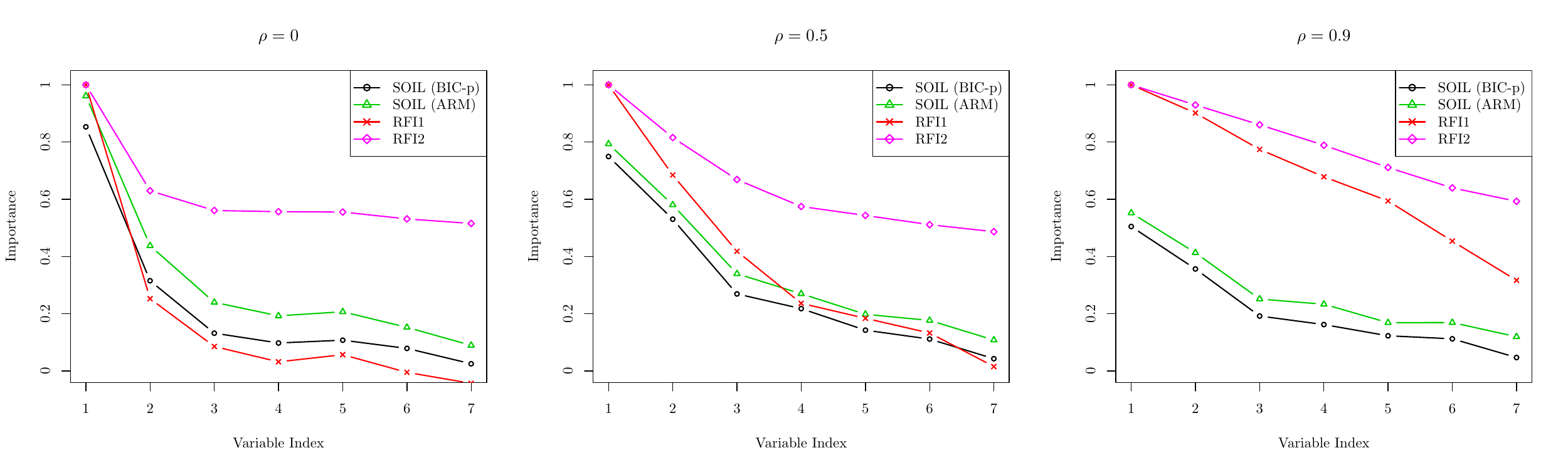}
\par\end{centering}

\caption{Simulation results for Example \ref{fig:res5}, where $n=80$, $p=6$. The true
coefficients $\boldsymbol{\beta}^{*}=\left(1,\frac{1}{2},\frac{1}{3},\frac{1}{4},\frac{1}{5},\frac{1}{6},0\right){}^{\intercal}$.\label{fig:res5}}
\end{figure}

\begin{figure}
\ra{0.8}
\begin{centering}
\includegraphics[scale=0.64]{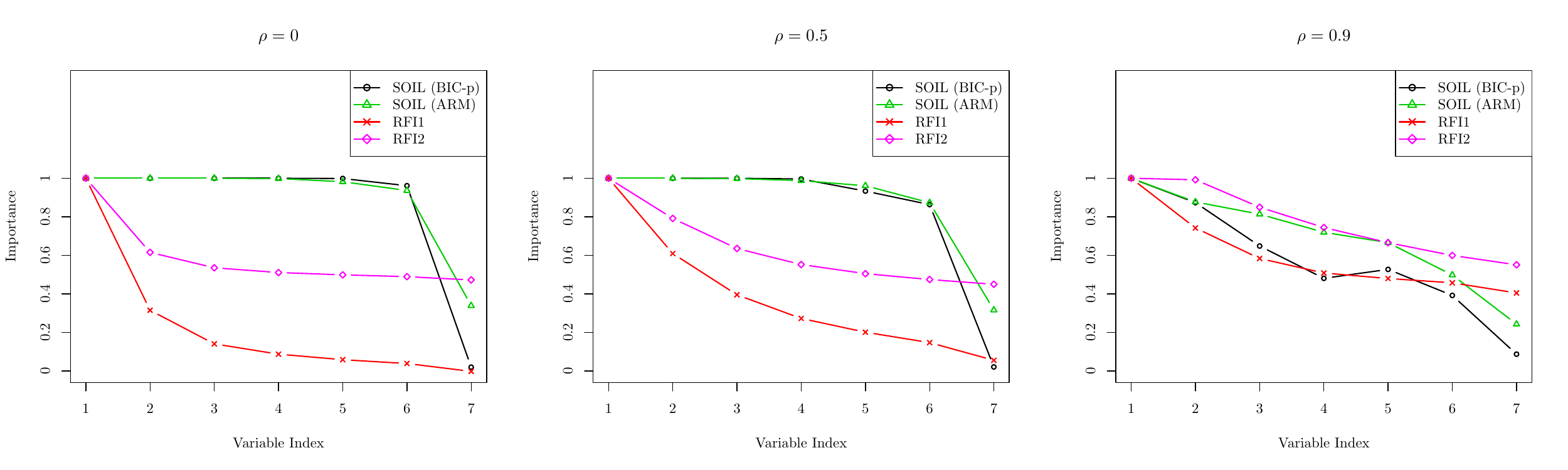}
\par\end{centering}

\caption{Simulation results for Example \ref{fig:res6}, where $n=5000$, $p=6$. The true
coefficients $\boldsymbol{\beta}^{*}=\left(1,\frac{1}{2},\frac{1}{3},\frac{1}{4},\frac{1}{5},\frac{1}{6},0\right){}^{\intercal}$.\label{fig:res6}}
\end{figure}

\subsection{Comparison with stability selection}
\citet{meinshausen2010stability} proposed a stability selection (SS) method to improve the Lasso variable selection. SS may be regarded as an importance measure. In Supplementary Materials Part C, we present a comparison of SS importance to our SOIL approach. Due to worse performances of SS compared with SOIL, together with the fact that
the main goal of SS is not on variable importance, we do not 
consider SS in our main simulation.

\section{REAL DATA EXAMPLES \label{sec:realdata}}

We apply the variable importance measures to two real datasets:

\subsubsection*{BGS data.}

We first consider a dataset with small $p$ from
the Berkeley Guidance Study (BGS) by \citet{tuddenham1954physical}.
The dataset includes 66 registered newborn boys whose physical growth
measures are followed for 18 years. Following \citet[p.179]{cook2009applied}
we consider a regression model of age 18 height on $p=6$ predictors:
weights at ages two (WT2) and nine (WT9), heights at ages two (HT2) and nine
(HT9), age nine leg circumference (LG9), and age 18 strength (ST18).
The corresponding SOIL-ARM, SOIL-BIC-p, LMG, RFI1 and RFI2 importances for each variable are computed and summarized
in Table \ref{fig:Importances-of-the BGS}. We found that HT9 is the
most important variable according to all methods. But different methods 
produce different second-most important variables.

\begin{table}[H]
\ra{0.8}
\begin{centering}
{\footnotesize{}{}}%
\begin{tabular}{lr@{\extracolsep{0pt}.}lr@{\extracolsep{0pt}.}lr@{\extracolsep{0pt}.}lr@{\extracolsep{0pt}.}lr@{\extracolsep{0pt}.}lr@{\extracolsep{0pt}.}l}
\toprule 
 & \multicolumn{2}{c}{{\footnotesize{}{}WT2}} & \multicolumn{2}{c}{{\footnotesize{}{}HT2}} & \multicolumn{2}{c}{{\footnotesize{}{}WT9}} & \multicolumn{2}{c}{{\footnotesize{}{}HT9}} & \multicolumn{2}{c}{{\footnotesize{}{}LG9}} & \multicolumn{2}{c}{{\footnotesize{}{}ST18}}\tabularnewline
\midrule 
{\footnotesize{}{}SOIL-ARM } & {\footnotesize{}{}0} & {\footnotesize{}{}16 } & {\footnotesize{}{}0} & {\footnotesize{}{}09 } & {\footnotesize{}{}0} & {\footnotesize{}{}03 } & \textbf{\footnotesize{}{}1} & \textbf{\footnotesize{}{}00}{\footnotesize{} } & \textbf{\footnotesize{}{}0} & \textbf{\footnotesize{}{}62}{\footnotesize{} } & {\footnotesize{}{}0} & {\footnotesize{}{}28}\tabularnewline
{\footnotesize{}{}SOIL-BIC-p } & {\footnotesize{}{}0} & {\footnotesize{}{}01 } & {\footnotesize{}{}0} & {\footnotesize{}{}00 } & {\footnotesize{}{}0} & {\footnotesize{}{}00 } & \textbf{\footnotesize{}{}1} & \textbf{\footnotesize{}{}00}{\footnotesize{} } & \textbf{\footnotesize{}{}0} & \textbf{\footnotesize{}{}63}{\footnotesize{} } & {\footnotesize{}{}0} & {\footnotesize{}{}08}\tabularnewline
{\footnotesize{}{}LMG } & {\footnotesize{}{}0} & {\footnotesize{}{}06 } & \textbf{\footnotesize{}{}0} & \textbf{\footnotesize{}{}13}{\footnotesize{} } & {\footnotesize{}{}0} & {\footnotesize{}{}08 } & \textbf{\footnotesize{}{}0} & \textbf{\footnotesize{}{}65}{\footnotesize{} } & {\footnotesize{}{}0} & {\footnotesize{}{}05 } & {\footnotesize{}{}0} & {\footnotesize{}{}02}\tabularnewline
{\footnotesize{}{}RFI1 } & {\footnotesize{}{}1} & {\footnotesize{}{}72 } & {\footnotesize{}{}2} & {\footnotesize{}{}50 } & {\footnotesize{}{}1} & {\footnotesize{}{}79 } & \textbf{\footnotesize{}{}55} & \textbf{\footnotesize{}{}66}{\footnotesize{} } & \textbf{\footnotesize{}{}4} & \textbf{\footnotesize{}{}12}{\footnotesize{} } & {\footnotesize{}{}1} & {\footnotesize{}{}05}\tabularnewline
{\footnotesize{}{}RFI2 } & {\footnotesize{}{}70} & {\footnotesize{}{}89 } & {\footnotesize{}{}101} & {\footnotesize{}{}58 } & {\footnotesize{}{}100} & {\footnotesize{}{}52 } & \textbf{\footnotesize{}{}2126} & \textbf{\footnotesize{}{}64}{\footnotesize{} } & {\footnotesize{}{}123} & {\footnotesize{}{}52 } & \textbf{\footnotesize{}{}127} & \textbf{\footnotesize{}{}74}\tabularnewline
\bottomrule
\end{tabular}
\par\end{centering}{\footnotesize \par}

\caption{Importances measures of the variables in BGS data. The top two most
important variables according to each measure are in bold. \label{fig:Importances-of-the BGS}}
\end{table}

Then we conduct a ``credibility check" for the above results of various importance measures. To do so we use a guided simulation or cross-examination \citep{li2000interactive,rolling2014model}, in which the performances of the importance measures are tested using data that are simulated from models recommended by the importance measures respectively. The basic idea of cross-examination is that one usually anticipates that a good method should have a better performance than other methods on the simulated data that are constructed from the method itself. In our context, if we compute the variable importances $S_{1}^{A}\ddd S_{p}^{A}$ on a real dataset using measure $A$, and construct a suggested model (with top rated important variables) and simulate a new dataset from this model, then on the new dataset, the variable importances $\tilde{S}_{1}^{A}\ddd \tilde{S}_{p}^{A}$ using measure $A$ should be more similar  to $S_{1}^{A}\ddd S_{p}^{A}$ than the variable importances $\tilde{S}_{1}^{B}\ddd \tilde{S}_{p}^{B}$ using measure $B$. Otherwise, one can naturally question the adequacy of applying measure $A$ to the original real data.

The cross-examination procedure is as follows:
\begin{enumerate}
\item Choose one measure from SOIL-ARM, SOIL-BIC-p, LMG, RFI1 and RFI2 as the base measure, and select the resulting top two most important variables (e.g. HT9 and LG9 if SOIL-ARM is the base measure).
\item Fit linear regression using only the selected variables as predictors, and obtain the estimated coefficients $\boldsymbol{\widehat{\beta}}$
and standard deviation $\widehat{\sigma}$. 
\item Generate the new response according to the model: ${\bf Y}_{new}=\mathbf{X}\boldsymbol{\widehat{\beta}}+\widehat{\sigma}N(0,1)$.
\item Compute the SOIL-ARM, SOIL-BIC-p, LMG, RFI1 and RFI2 importance measures using the new dataset $({\bf X},{\bf Y}_{new})$.
\item Repeat the above steps 100 times and take the average of each importance.
\item Go to Step 1 until all measures have served as the base measure.
\end{enumerate}
\begin{figure}
\begin{centering}
\includegraphics[scale=0.60]{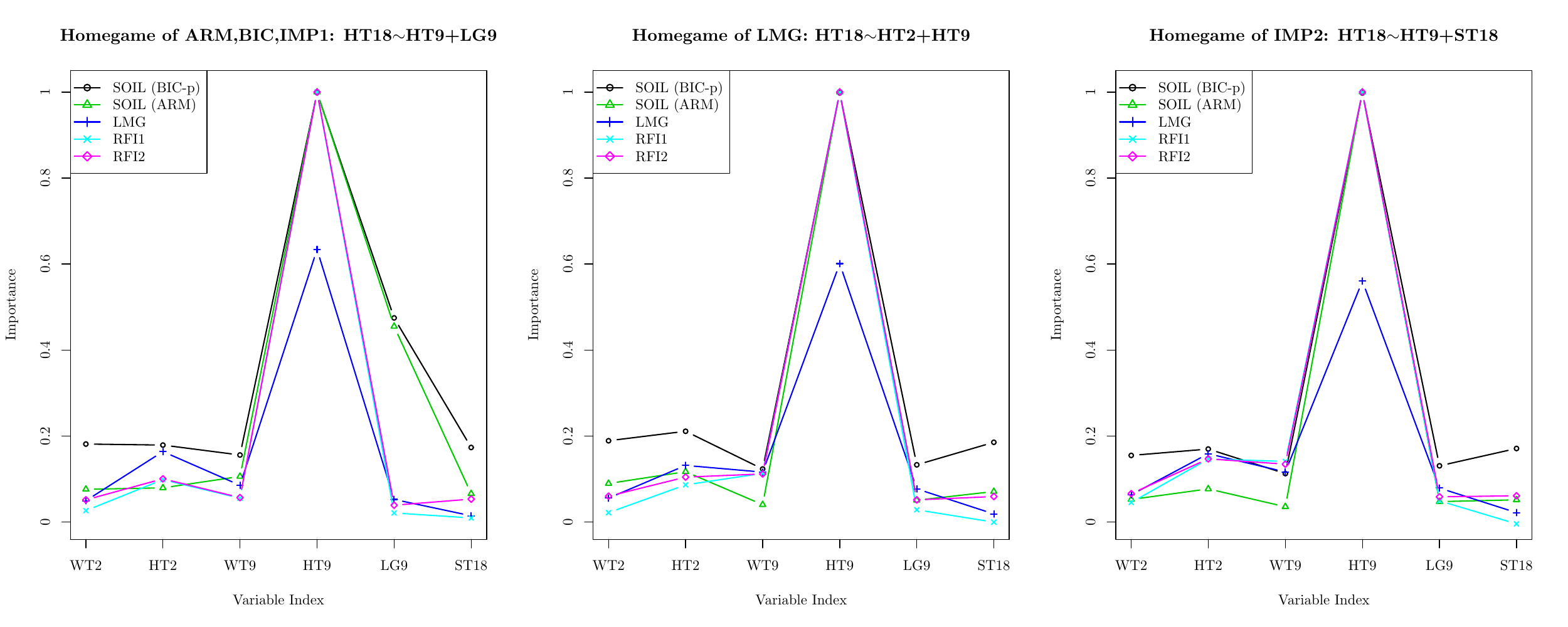}
\par\end{centering}

\caption{Results of cross-examination for BGS data.\label{fig:result of guided simu BGS}}
\end{figure}

The results are depicted in Figure \ref{fig:result of guided simu BGS}.
Overall, SOIL-ARM and SOIL-BIC-p perform reasonably better than the other importance
measures. In the home-game (when the variable is selected as the base measure) of SOIL-ARM, SOIL-BIC-p and RFI1, we can
see that LMG and random forest (RFI1 or RFI2) do not support the true
variable LG9, while SOIL-ARM or SOIL-BIC-p clearly indicate, correctly,
HT9 and LG9 as the important ones (although with less confidence on LG9). In fact, LMG, RFI1 and RFI2 all view HT2 as more important than LG9, a mistake seemingly caused by the higher correlation of HT2 (0.57) to HT18 than LG9 (0.37). In the home-game of LMG, all
methods single out only HT9 as the most important (but not HT2). However, SOIL-ARM and SOIL-BIC-p
assign the second largest importance to HT2, which is consistent with
the aforementioned Order Preserving property. The random forest importance
measures do not show this property. The home-game of
RFI2 is similar to the home-game of LMG, where the Order Preserving
property still holds for SOIL-ARM and SOIL-BIC-p but not for the others.

\subsubsection*{Bardet data.}

For a dataset with large $p$, we consider the Bardet dataset. It collects tissue samples from the eyes of 120 twelve-week-old male rats, which are the offspring of inter-crossed F1 animals. For
each tissue, the RNAs of 31,042 selected probes are measured by the
normalized intensity valued. The gene intensity values are in log
scale. To investigate the genes that are related to gene TRIM32,
which causes the Bardet-Biedl syndrome according to \citet{chiang2006homozygosity}, a screening method \citep{huang2008adaptive} is applied to the original probes,
which gives us a dataset with 200 probes for each of 120 tissues.
We use this screened dataset to carry out our importance measure analysis.

Since LMG is not feasible to handle cases with $p>20$, it is not included in our analysis below. The corresponding SOIL-ARM, SOIL-BIC-p, RFI1 and RFI2 importances for most relevant variable are summarized in Table \ref{fig:Top-10-genes}.   We present the top ten variables according to the different importance measures respectively. The name of each gene is too long, so for convenience
we record the corresponding EST number instead. From Table \ref{fig:Top-10-genes},
we can see that different importance measures have very different results.

\begin{table}[H]
\ra{0.8}
\begin{centering}
{\footnotesize{}{}}%
\begin{tabular}{cllllllll}
\toprule 
{\footnotesize{}{}Rank } & \multicolumn{2}{l}{{\footnotesize{}{}ARM}} & \multicolumn{2}{l}{{\footnotesize{}{}BIC-p}} & \multicolumn{2}{l}{{\footnotesize{}{}RFI1}} & \multicolumn{2}{l}{{\footnotesize{}{}RFI2}}\tabularnewline
\midrule 
{\footnotesize{}{}1 } & \textbf{\footnotesize{}{}25141}{\footnotesize{} } & {\footnotesize{}{}1.000 } & \textbf{\footnotesize{}{}25141}{\footnotesize{} } & {\footnotesize{}{}1.000 } & \textbf{\footnotesize{}{}25141}{\footnotesize{} } & {\footnotesize{}{}5.113 } & \textbf{\footnotesize{}{}21907}{\footnotesize{} } & {\footnotesize{}{}0.061 }\tabularnewline
{\footnotesize{}{}2 } & \textbf{\footnotesize{}{}28967}{\footnotesize{} } & {\footnotesize{}{}0.935 } & \textbf{\footnotesize{}{}28967}{\footnotesize{} } & {\footnotesize{}{}1.000 } & \textbf{\footnotesize{}{}21907}{\footnotesize{} } & {\footnotesize{}{}5.006 } & \textbf{\footnotesize{}{}25141}{\footnotesize{} } & {\footnotesize{}{}0.059 }\tabularnewline
{\footnotesize{}{}3 } & \textbf{\footnotesize{}{}28680}{\footnotesize{} } & {\footnotesize{}{}0.834 } & \textbf{\footnotesize{}{}28680}{\footnotesize{} } & {\footnotesize{}{}0.999 } & \textbf{\footnotesize{}{}11711}{\footnotesize{} } & {\footnotesize{}{}4.875 } & \textbf{\footnotesize{}{}11711}{\footnotesize{} } & {\footnotesize{}{}0.054 }\tabularnewline
{\footnotesize{}{}4 } & \textbf{\footnotesize{}{}30141}{\footnotesize{} } & {\footnotesize{}{}0.576 } & \textbf{\footnotesize{}{}30141}{\footnotesize{} } & {\footnotesize{}{}0.491 } & \textbf{\footnotesize{}{}11719}{\footnotesize{} } & {\footnotesize{}{}4.778 } & \textbf{\footnotesize{}{}25105}{\footnotesize{} } & {\footnotesize{}{}0.041 }\tabularnewline
{\footnotesize{}{}5 } & {\footnotesize{}{}21092 } & {\footnotesize{}{}0.397 } & {\footnotesize{}{}21092 } & {\footnotesize{}{}0.278 } & \textbf{\footnotesize{}{}25105}{\footnotesize{} } & {\footnotesize{}{}4.491 } & \textbf{\footnotesize{}{}24565}{\footnotesize{} } & {\footnotesize{}{}0.036 }\tabularnewline
{\footnotesize{}{}6 } & {\footnotesize{}{}15863 } & {\footnotesize{}{}0.261 } & {\footnotesize{}{}15863 } & {\footnotesize{}{}0.142 } & \textbf{\footnotesize{}{}9303}{\footnotesize{} } & {\footnotesize{}{}4.332 } & \textbf{\footnotesize{}{}28680}{\footnotesize{} } & {\footnotesize{}{}0.035 }\tabularnewline
{\footnotesize{}{}7 } & {\footnotesize{}{}17599 } & {\footnotesize{}{}0.219 } & {\footnotesize{}{}17599 } & {\footnotesize{}{}0.121 } & \textbf{\footnotesize{}{}28680}{\footnotesize{} } & {\footnotesize{}{}4.239 } & \textbf{\footnotesize{}{}25403}{\footnotesize{} } & {\footnotesize{}{}0.034 }\tabularnewline
{\footnotesize{}{}8 } & {\footnotesize{}{}22813 } & {\footnotesize{}{}0.106 } & {\footnotesize{}{}25367 } & {\footnotesize{}{}0.028 } & \textbf{\footnotesize{}{}25425}{\footnotesize{} } & {\footnotesize{}{}3.788 } & \textbf{\footnotesize{}{}9303}{\footnotesize{} } & {\footnotesize{}{}0.033 }\tabularnewline
{\footnotesize{}{}9 } & {\footnotesize{}{}25367 } & {\footnotesize{}{}0.079 } & {\footnotesize{}{}22813 } & {\footnotesize{}{}0.016 } & \textbf{\footnotesize{}{}16569}{\footnotesize{} } & {\footnotesize{}{}3.733 } & \textbf{\footnotesize{}{}22029}{\footnotesize{} } & {\footnotesize{}{}0.032 }\tabularnewline
{\footnotesize{}{}10 } & {\footnotesize{}{}24892 } & {\footnotesize{}{}0.047 } & {\footnotesize{}{}14949 } & {\footnotesize{}{}0.005 } & \textbf{\footnotesize{}{}22029}{\footnotesize{} } & {\footnotesize{}{}3.680 } & \textbf{\footnotesize{}{}24087}{\footnotesize{} } & {\footnotesize{}{}0.030 }\tabularnewline
\bottomrule
\end{tabular}
\par\end{centering}{\footnotesize \par}

\centering{}\caption{Top ten genes for different variable importance measures for Bardet
data.\label{fig:Top-10-genes}}
\end{table}

Notice that $X_{25141}$ is the most important variable according
to Table \ref{fig:Top-10-genes}. Random forest is unstable in the sense that each
time we compute the random forest importance on the data, the top ten variables obtained
tended to be quite different in terms of their rankings. For SOIL-BIC-p and SOIL-ARM,
the top four genes always have the same rank and the importance values are pretty much the same in different runs. Also, a striking feature for the random forest in this data example is that the values of the importances
are quite close to each other and decaying gradually, making it hard to judge which variables are really important. 

We carry out a guided simulation study similar to that for the BGS data, except that LMG is not included. Based on the information in Table \ref{fig:Top-10-genes}, the top 4 variables are selected for SOIL-BIC-p (SOIL-ARM), and the top 10 for RFI1 and RFI2 respectively. 

In Figure \ref{fig:Simu_colon}, we only present the variable importances of the ``true" genes due to space limitation. RFI1 and RFI2 are all normalized. In the home-game of SOIL-ARM and SOIL-BIC-p, both can correctly select
all the true variables if  the cut-off value is set at $0.4$. For random forest, however, the maximum RFI1 and RFI2 values among the unimportant ones exceed the most important
ones respectively, indicating that the random forest has difficulty differentiating the really important and unimportant variables.  

In the home-game of RFI1 and RFI2, none of the competitors performs very well. With the generating model being larger, with the limited information in the data (in conjunction with the complicated correlation
among the genes), the importance measures simply cannot reveal all
the true variables. Only the true variable $X_{25414}$
is differentiated clearly by all methods. From the SOIL perspective, it is willing to support at most 3 more variables with some confidence. Random forest gives more true variables significant importance values. A drawback is that some noise variables receive relatively large importance values, which are even higher than almost half of the true variables.  

From the guided simulations,  the Order Preserving property fails in all the cases for the random forest importance measures. For SOIL, in the home-game of ARM
and BIC-p, it holds for both SOIL-ARM and SOIL-BIC-p; but in the home-game of RFI1 or RFI2, the property does not  hold exactly, but it does hold in the sense that the maximum importance of the noise variables is still very small (and it is not meaningful to rank the variables with tiny importance values). The key point here is that while SOIL certainly can miss subtle variables in the true model when the sample size is small,  it typically does not recommend an unimportant variable as important. The same cannot be said for the other importance measures.    
\begin{figure}
\begin{centering}
\includegraphics[scale=0.6]{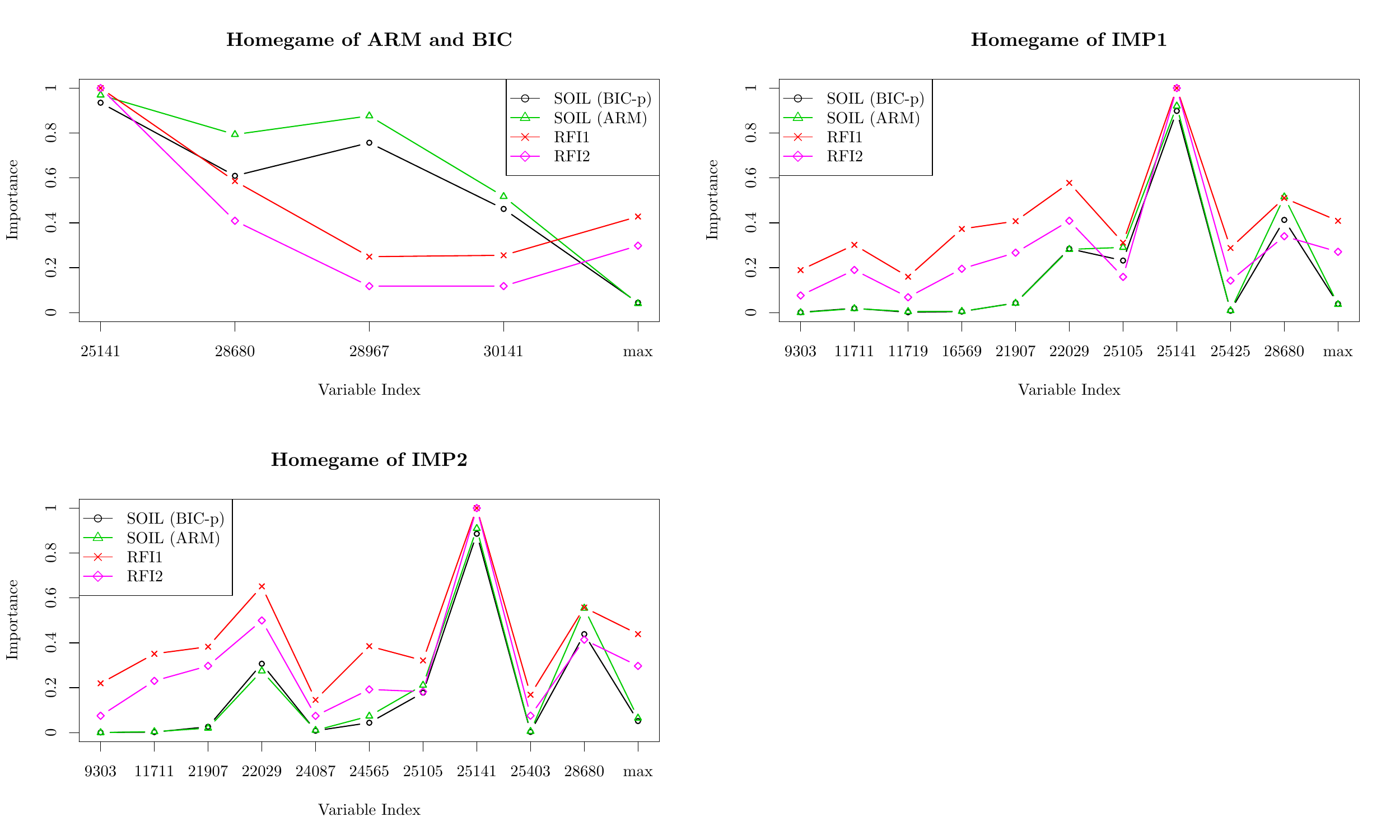}
\par\end{centering}

\caption{Simulation results for cross-examination\label{fig:Simu_colon}}
\end{figure}

\section{CONCLUSION AND DISCUSSION}

Variable importance is aimed to find the important variables for explanation or prediction of the response. 
The motivation is most natural but the task of devising an importance measure is quite tricky. Several challenges immediately arrive: 1. Importance depends on the goal of the analysis and application. Different goals may require different importance measures. 2. Should importance be based on parametric models or nonparametric models? Both seem to be valuable in our view. 3. Should the importance measure be purely relative to compare different variables or should their values have some meaning on their own?

The topic is even controversial, with attitude ranging from enthusiasm in research and/or application, to reluctant acceptance as a practical approach to deal with many predictors, to total pessimism on the topic that dismisses the possibility of general successes. The different opinions are all valid, properly reflecting the complexity and multi-facet nature of the problem.

In our opinion, there are two important facts to keep in mind. One is that people crave for importance measures, love ranking, and they put them in use. This calls for more research on the topic. The other is the currently still dominating practice of ``winner-takes-all", which is definitely a culprit of irreproducibility of many research results.  For reasonably complex data, making inference and decision based on a final selected model can lead to severely biased conclusions. A reliable importance measure can provide much needed complementary information to that from a final model and substantially improve the reliability of data analysis.            

We have investigated the variable
importance in linear regression and classification cases. The proposed new variable
importance measure (SOIL) is driven by model combination for considering
more than a single model, thus giving us an understanding of all the
variables, instead of only the ``important" ones in view of a single model. It is seen from
both the simulation results and the real data examples that the
SOIL approach has several desirable features such as exclusion/inclusion,
order preserving and robustness in several aspects, and performs very well compared to other variable importance measures considered. 

As \citet{gromping2015variable} pointed out in her paper, there is
no commonly accepted theoretical framework in the variable importance area. Not surprisingly, many critiques
on variable importance measures come up. \citet{ehrenberg1990unimportance}
pointed out that one should focus on the underneath causal mechanism
instead of the relative importance. We think SOIL is satisfactory in this regard. First, given enough information, SOIL assigns variable importance
close to one for these true predictors, which is consistent with revealing
the causal relationship between the response and the predictors. Second,
the SOIL importance of a variable goes beyond relative assessment of the variables and it gives an absolute sense on how much a variable is needed in the linear modeling with the available information. In
regression settings,  data analysts often use $t$ statistic or $p$-value to see if a variable
is significant or not. \citet{kruskal1989concepts}
pointed out that this pertains to a different concept. In their view, variable importance
is a population property while significance is a property of both
population and sample. To us, since all models are only approximations to model the data, 
there is advantage to treat variable importance measures as data dependent quantities that reflect the nature of the data. SOIL intends to do just that. 

Note that the two importance measures by the random forecast are not based on parametric modeling. When the GLM framework does not work for the data, our SOIL approach may not provide valuable information while random forest based ones may.

To be fair, it may be debatable if a variable that has some predictive power (one way or another) but is not needed in the best model should be given significant (reasonably strong) importance or not. Our view is that it seems rare to consider the covariates only individually and thus it is better to reflect the goal of finding the best set of covariates to explain the response in the importance measures. From this angle, while giving out relevant variables is certainly useful, it may not be the most essential from a modeling perspective.

Through our simulation work, we have shown that the other methods often give clearly higher importance to variables that are not in the true model and/or give lower values for some variables in the true model when the covariates are correlated, error variance is large, or there are interaction terms. In real applications, these situations occur rather commonly. Thus the results seem to suggest that when sparse modeling is the goal, those importance measures may not directly provide objective variable assessment information.

\section*{APPENDIX}
\subsection*{Proof of Theorem \ref{thm:Under-the-assumption}.}
\begin{proof}
Denote by $\as\backslash\ak$
the set of variables contained in $\as$ but not in $\ak$. Since
\begin{eqnarray*}
\dfrac{\sum_{k=1}^{K}w_{k}|\as\backslash\ak|}{r^{*}} & = & \dfrac{\sum_{k=1}^{K}w_{k}\sum_{j\in\mathcal{A}^{*}}I(j\notin\ak)}{r^{*}}\\
 & = & \dfrac{\sum_{j\in\mathcal{A}^{*}}\sum_{k=1}^{K}w_{k}I(j\notin\ak)}{r^{*}}\\
 & = & \dfrac{\sum_{j\in\mathcal{A}^{*}}\sum_{k=1}^{K}w_{k}(1-I(j\in\ak))}{r^{*}}\\
 & = & \dfrac{\sum_{j\in\mathcal{A}^{*}}(1-S_{j})}{r^{*}}.
\end{eqnarray*}
and by the definition of weak consistency,
\[
0\leq\dfrac{\sum_{k=1}^{K}w_{k}|\as\backslash\ak|}{r^{*}}\leq\dfrac{\sum_{k=1}^{K}w_{k}|\ak\nabla\as|}{r^{*}}{\displaystyle \ \overset{p}{\to}\ }0.
\]
Hence,
\begin{eqnarray*}
\dfrac{\sum_{j\in\mathcal{A}^{*}}(1-S_{j})}{r^{*}} & {\displaystyle \overset{p}{\to}} & 0.
\end{eqnarray*}
On the other hand, 
\begin{eqnarray*}
\dfrac{\sum_{j\notin\mathcal{A}^{*}}S_{j}}{r^{*}} & = & \dfrac{\sum_{j\notin\mathcal{A}^{*}}\sum_{k=1}^{K}w_{k}I(j\in\ak)}{r^{*}}\\
 & = & \dfrac{\sum_{k=1}^{K}w_{k}\sum_{j\notin\mathcal{A}^{*}}I(j\in\ak)}{r^{*}}\\
 & = & \dfrac{\sum_{k=1}^{K}w_{k}|\ak\backslash\as|}{r^{*}}\\
 & \leq & \dfrac{\sum_{k=1}^{K}w_{k}|\ak\nabla\as|}{r^{*}}\overset{p}{\to}0.
\end{eqnarray*}

\end{proof}

\begin{proof}
[Proof of Theorem \ref{thm:Let-,-}] Assume $\dfrac{|\overline{\aa}_{c}|}{r^{*}}$
does not converge to $0$ in probability as $n$ tends to infinity
($r^{*}$ may or may not depend on $n$), then we have a subsequence
, which for convenience we still denote by $\dfrac{|\overline{\aa}_{c}|}{r^{*}}$,
that is greater than a non-zero positive constant $\epsilon_{0}$,
i.e. $\dfrac{|\overline{\aa}_{c}|}{r^{*}}\geq\epsilon_{0}$ in probability.
Thus,

\begin{eqnarray*}
\dfrac{\sum_{j\in\mathcal{A}^{*}}(1-S_{j})}{r^{*}} & = & \dfrac{\sum_{j\in\as,S_{j}\leq c}(1-S_{j})}{r^{*}}+\dfrac{\sum_{j\in\as,S_{j}>c}(1-S_{j})}{r^{*}}\\
 & \geq & \dfrac{\sum_{j\in\as,S_{j}\leq c}(1-S_{j})}{r^{*}}\\
 & \geq & \dfrac{\sum_{j\in\as,S_{j}\leq c}(1-c)}{r^{*}}\\
 & = & (1-c)\dfrac{\sum_{j\in\mathcal{A}^{*}}I(S_{j}\leq c)}{r^{*}}\\
 & = & (1-c)\dfrac{|\overline{\aa}_{c}|}{r^{*}}\\
 & \geq & (1-c)\epsilon_{0},
\end{eqnarray*}
which contradicts with Theorem \ref{thm:Let-,-}. Hence, we have $\dfrac{|\overline{\aa}_{c}|}{r^{*}}{\displaystyle \ \overset{p}{\to}\ }0.$
Similarly, we can prove $\dfrac{|\underline{\aa}_{c}|}{r^{*}}{\displaystyle \ \overset{p}{\to}\ }0$.
\end{proof}

\bigskip
\section*{SUPPLEMENTARY MATERIALS}

\begin{description}

\item[R-package for SOIL:] R-package \texttt{SOIL} containing code to compute the SOIL importance measure described in the article. (GNU zipped tar file)

\item[Real data sets:] Data sets BGS and Bardet used in the illustration of SOIL in Section~\ref{sec:realdata}. (.rda file)

\item[Text document:] Supplementary materials for ``Sparsity Oriented Importance Learning for High-dimensional Linear Regression'' by Chenglong Ye, Yi Yang and Yuhong Yang. (.pdf file)

\end{description}

\bibliographystyle{agsm}
\bibliography{extracted}

\pagebreak

\begin{center}
\textbf{\large Supplemental Materials for ``Sparsity Oriented Importance Learning for High-dimensional Linear Regression"}
\end{center}

\setcounter{table}{0}
\renewcommand{\thetable}{S\arabic{table}}

\setcounter{figure}{0}
\renewcommand{\thefigure}{S\arabic{figure}}

\subsection*{Part A: Weighting using generalized fiducial inference.}

Based on Fisher's controversial fiducial idea, \citet{lai2015generalized} proposed the generalized fiducial inference applied to ``large $p$ small $n$'' problem.  Their paper concerns the generalized fiducial inference for the linear regression case. For each candidate model $\ak$,
the fiducial probability for the model is 
\[
p(\ak)\propto R(\ak)\equiv\Gamma(\frac{n-\left|\ak\right|}{2})(\pi RSS_{\ak})^{-\frac{n-\left|\ak\right|-1}{2}}n^{-\frac{\left|\ak\right|+1}{2}}\left(\begin{array}{c}
p\\
\left|\ak\right| 
\end{array}\right)^{-\gamma},
\]
where $RSS_{\ak}$ is the residual sum of squares of
$\ak$. For a practical reason, the authors approximate the above
fiducial probability by 
\[
r(\ak)\approx R(\ak)/\sum_{l=1}^{K}R(\mathcal{A}^{l}).
\]
We can use $r(\ak)$ as the weight $w_{k}$ for each candidate model.
It is shown in their paper that the true model will have the highest
fiducial probability among all the candidate models.

\subsection*{Part B: Additional simulation results.}
In this part, we provide the results of Example \ref{fig:resS1}-\ref{fig:resS5}, whose settings are described in Table \ref{tab:simu settings} of the main body of the article. These results support our conclusions as discussed in Section \ref{sec:property}.
\begin{figure}[H]
\begin{centering}
\includegraphics[scale=0.64]{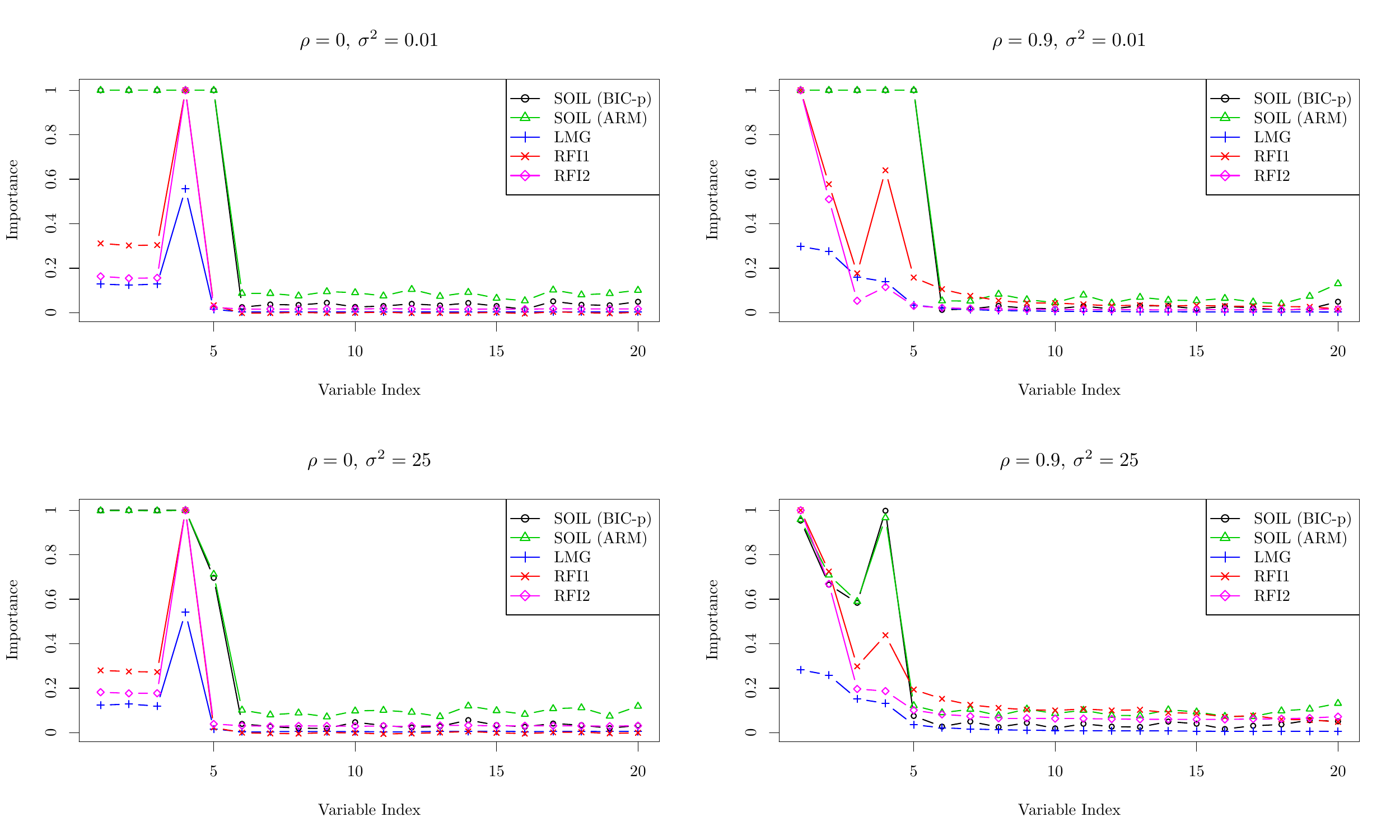}
\par\end{centering}
\caption{Simulation results for Example \ref{fig:resS1}, where $n=150$, $p=20$. The true
coefficients $\boldsymbol{\beta}^{*}=(4,4,4,-6\sqrt{2},\frac{3}{4},0,...,0)$.\label{fig:resS1}}
\end{figure}

\begin{figure}[H]
\begin{centering}
\includegraphics[scale=0.64]{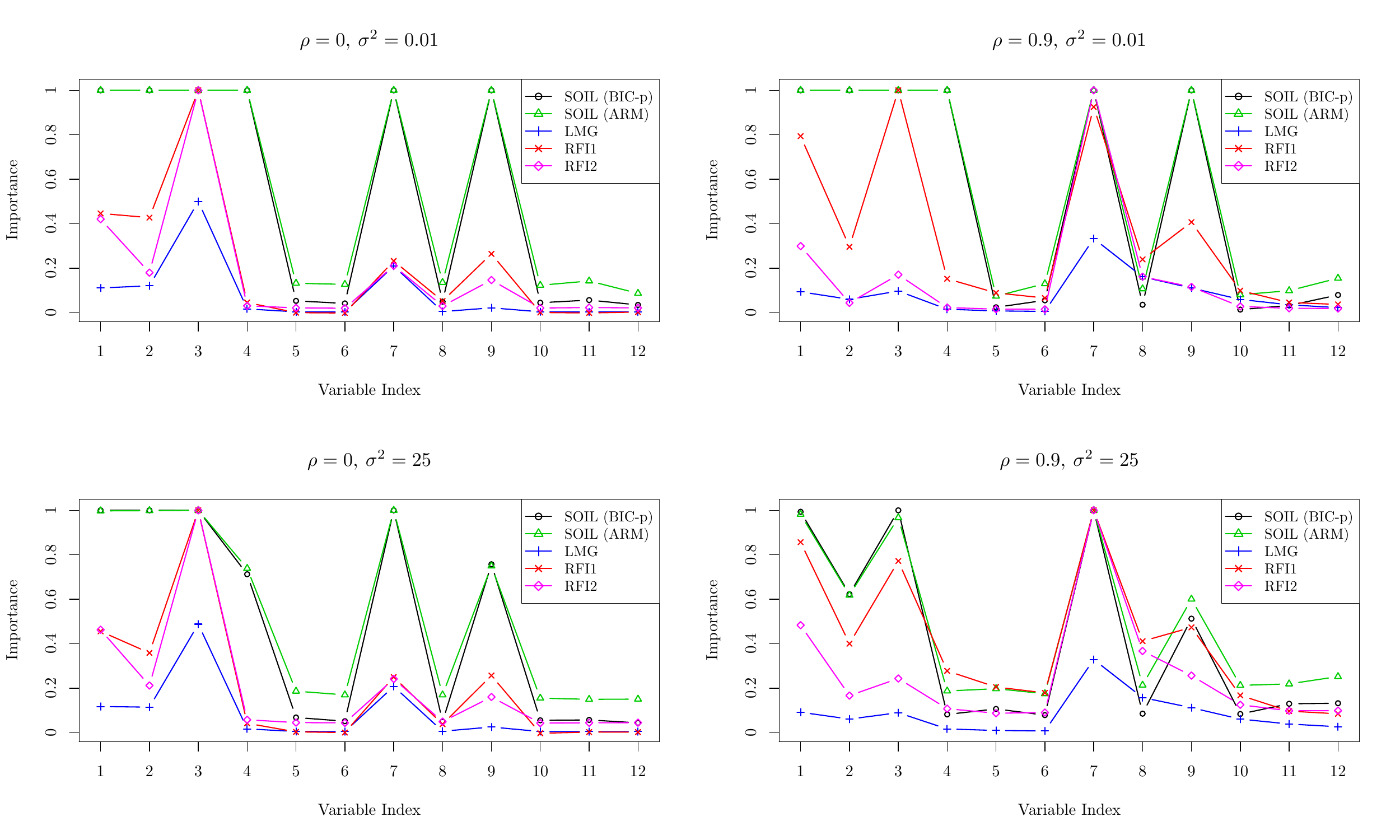}
\par\end{centering}

\caption{Simulation results for Example \ref{fig:resS2}, where $n=150$, $p=6$. The true
coefficients $\boldsymbol{\beta}^{*}=(4,4,-6\sqrt{2},\frac{3}{4},0,0)^{\intercal}$.
Add $(X_{1}^{2},X_{2}^{2},X_{3}^{2},X_{4}^{2},X_{5}^{2},X_{6}^{2})$
and corresponding coefficients $(\beta_{7}^{*},\beta_{8}^{*}\protect\ddd\beta_{12}^{*})^{\intercal}=(4,0,1,0,0,0)^{\intercal}$.\label{fig:resS2}}
\end{figure}

\begin{figure}[H]
\begin{centering}
\includegraphics[scale=0.64]{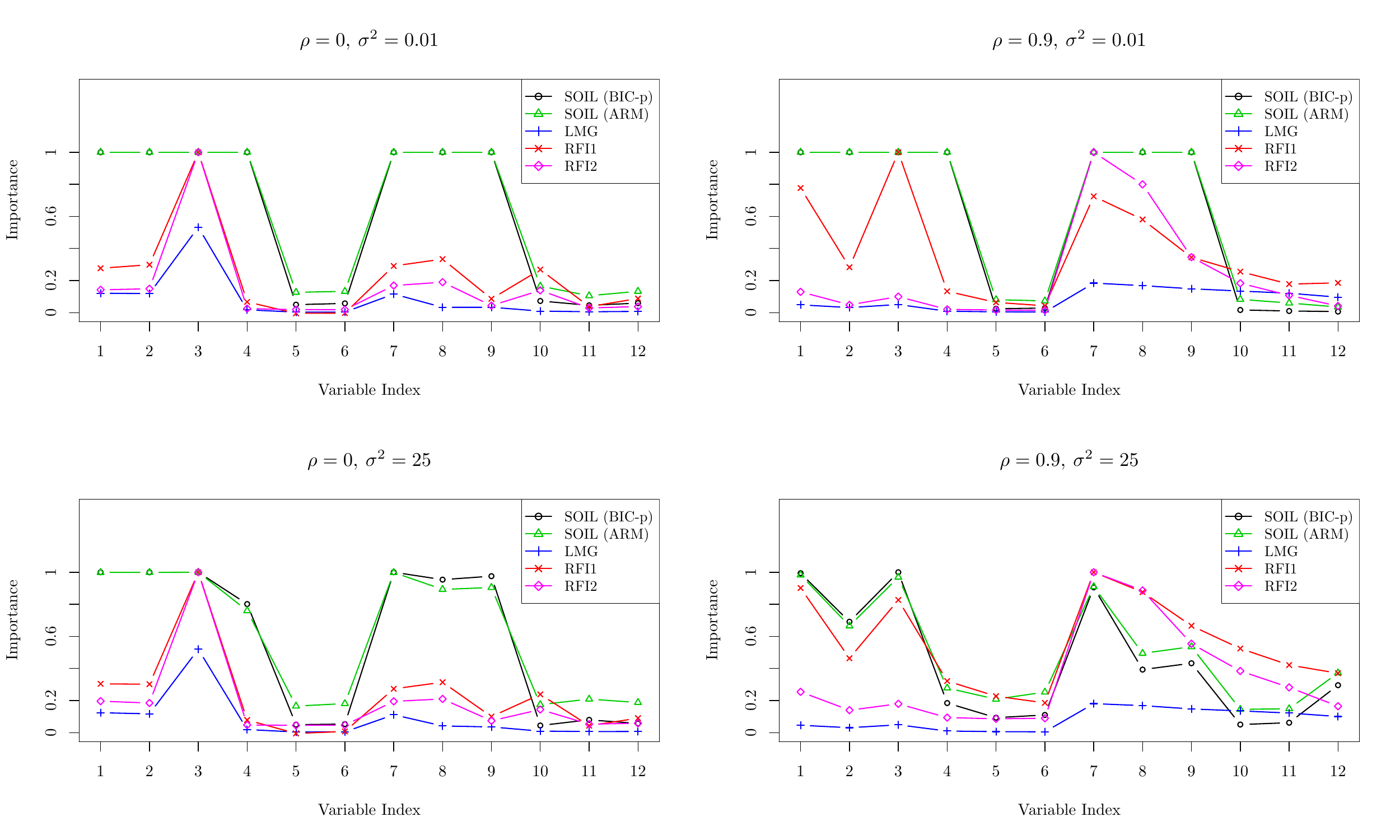}
\par\end{centering}

\caption{Simulation results for Example \ref{fig:resS3}, where $n=150$, $p=6$. The true
coefficient $\boldsymbol{\beta}^{*}=(4,4,-6\sqrt{2},\frac{3}{4},0,0)^{\intercal}$.
Add $(X_{1}X_{2},X_{1}X_{3},X_{1}X_{4},X_{2}X_{3},X_{2}X_{4},X_{3}X_{4})$
and corresponding coefficients $(\beta_{7}^{*},\beta_{8}^{*}\protect\ddd\beta_{12}^{*})^{\intercal}=(4,2,2,0,0,0)^{\intercal}$.\label{fig:resS3}}
\end{figure}

\begin{figure}[H]
\begin{centering}
\includegraphics[scale=0.64]{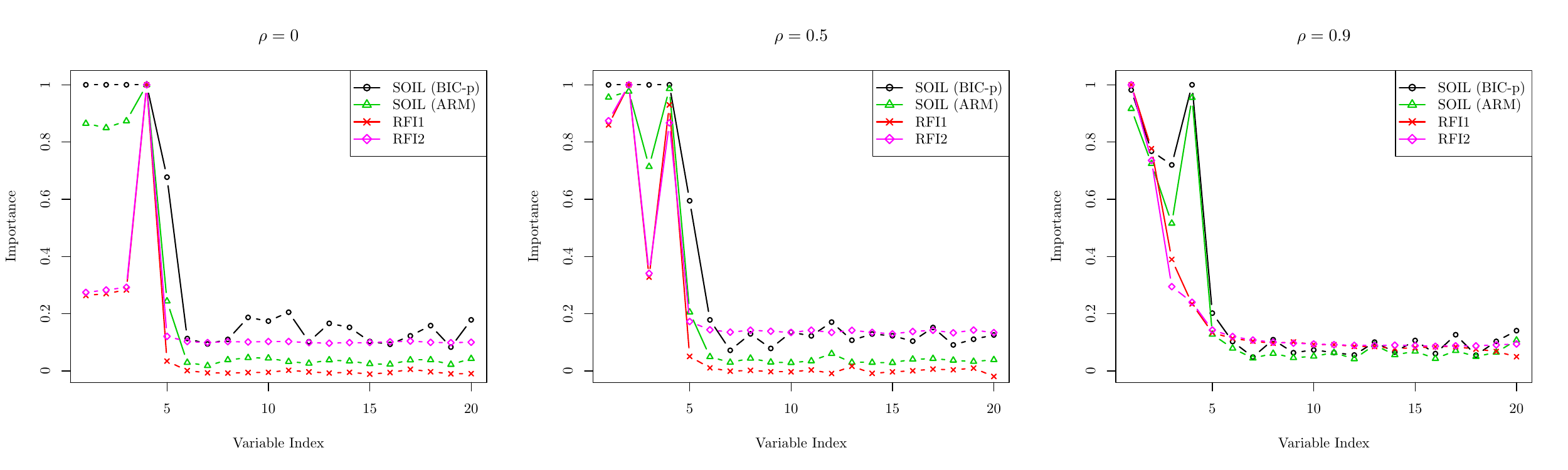}
\par\end{centering}

\caption{Simulation results for Example \ref{fig:resS4}, where $n=150$, $p=20$. The true
coefficients $\boldsymbol{\beta}^{*}=(4,4,4,-6\sqrt{2},\frac{3}{4},0,...,0)$.\label{fig:resS4}}
\end{figure}

\begin{figure}[H]
\begin{centering}
\includegraphics[scale=0.64]{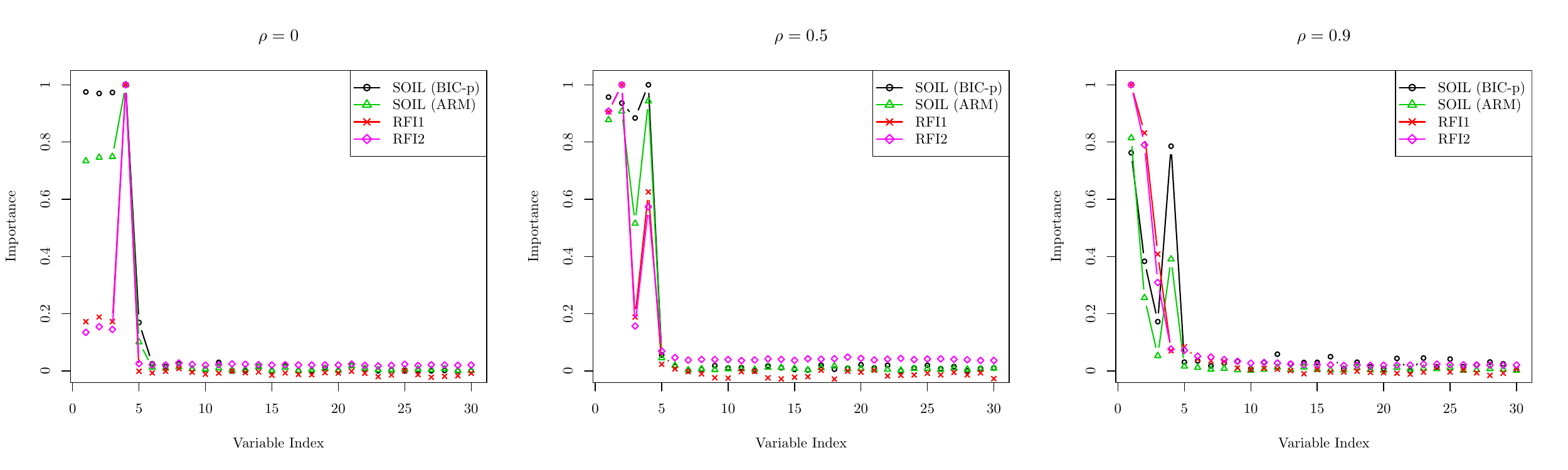}
\par\end{centering}

\caption{Simulation results for Example \ref{fig:resS5}, where $n=100$, $p=200$. The true
coefficients $\boldsymbol{\beta}^{*}=(4,4,4,-6\sqrt{2},\frac{3}{4},0,...,0)$.\label{fig:resS5}}
\end{figure}

\subsection*{Part C: Comparison with stability selection.}
In this subsection, we present a comparison of SS  \citep{meinshausen2010stability} importance to our SOIL importance.

The simulation data $\{y_{i},\mathbf{x}_{i}\}_{i=1}^{n}$ is generated
from the linear model $y_{i}=\mathbf{x}_{i}^{\intercal}\boldsymbol{\beta}^{*}+\epsilon_{i}$,
$\epsilon\sim N(0,\sigma^{2})$. We generate $\mathbf{x}_{i}$ from
multivariate normal distribution $N_{p}(0,\Sigma)$. For each element
$\Sigma_{ij}$ of $\Sigma$, $\Sigma_{ij}=\rho^{\left|i-j\right|}$,
i.e. the correlation of $X_{i}$ and $X_{j}$ is $\rho^{\left|i-j\right|}$.
We consider two cases, the settings of which are listed in Table \ref{tab:Simulation-settingssssss}.

\begin{table}[H]
\begin{centering}
\begin{tabular}{cccccc}
\toprule 
Example & $n$ & $p$ & $\rho$ & $\sigma^{2}$ & Coefficients\tabularnewline
\hline 

1 & 100 & 20 & 0 & 0.01 & {\small{}$\boldsymbol{\beta}^{*}=(4,4,4,-6\sqrt{2},\frac{3}{4},0,...,0)^{\intercal}$}\tabularnewline
 
2 & 100 & 20 & 0.7 & 0.1 & {\small{}$\boldsymbol{\beta}^{*}=(4,0,4,-6\sqrt{2},\frac{3}{4},0,...,0)^{\intercal}$}\tabularnewline
\bottomrule 
\end{tabular}
\par\end{centering}

\caption{Simulation settings for  SS \label{tab:Simulation-settingssssss}}
\end{table}

It can be seen from Tables \ref{tab:Variable-importance-foreg12} and \ref{tab:Variable-importance-foreg13} that SS
does not give enough importance to the true variable $X_{5}$ in Example 1 while it more strongly supports the noise variable $X_{2}$ than the true variable $X_{5}$ in Example 2, which leads to unavoidable incorrect variable selection regardless of the cutoff to be used to decide if a variable is in or out based on its importance. In contrast, SOIL-ARM and SOIL-BIC-p pick all
the important variables and leave noise variables out. From these results, together with the fact that
the main goal of SS is not on variable importance, we have not 
considered stability selection in the main simulations in this work. 

\begin{table}[H]
\begin{centering}
\begin{tabular}{lcccccc}
\toprule
Method/Variable  & $X_{1}$ & $X_{2}$ & $X_{3}$ & $X_{4}$ & $X_{5}$ & max of rest\tabularnewline
\hline 
 
SOIL-ARM & 1.00 & 1.00 & 1.00 & 1.00 & 1.00 & 0.12\tabularnewline
 
SOIL-BIC-p & 1.00 & 1.00 & 1.00 & 1.00 & 1.00 & 0.07\tabularnewline
 
Stability Selection & 0.99 & 0.99 & 0.99 & 1.00 & 0.02 & 0.002\tabularnewline
\bottomrule 
\end{tabular}
\par\end{centering}

\caption{Variable importance for Example 1. \label{tab:Variable-importance-foreg12}}
\end{table}

\begin{table}[H]
\begin{centering}
\begin{tabular}{lcccccc}
\toprule  
Method/Variable  & $X_{1}$ & $X_{2}$ & $X_{3}$ & $X_{4}$ & $X_{5}$ & max of rest\tabularnewline
\hline  
SOIL-ARM & 1.00 & 0.15 & 1.00 & 1.00 & 1.00 & 0.14\tabularnewline
 
SOIL-BIC-p & 1.00 & 0.06 & 1.00 & 1.00 & 1.00 & 0.05\tabularnewline
 
Stability Selection & 1.00 & 0.44 & 0.94 & 1.00 & 0.26 & 0.05\tabularnewline
\bottomrule 
\end{tabular}
\par\end{centering}

\caption{Variable importance for Example 2.\label{tab:Variable-importance-foreg13}}
\end{table}

\end{document}